\documentclass[letterpaper,twocolumn,prl,aps,superscriptaddress,amsmath,amssymb,floatfix]{revtex4-2}
\usepackage{mathptmx}
\usepackage[latin9]{inputenc}
\usepackage{color}
\usepackage{upgreek}
\usepackage{amsmath}
\usepackage{amssymb}
\usepackage{calrsfs}
\DeclareMathAlphabet{\mathcal}{OMS}{zplm}{m}{n}
\SetMathAlphabet\mathcal{bold}{OMS}{zplm}{bx}{n}
\usepackage{graphicx}
\usepackage{esint}
\usepackage{bm}
\usepackage{bbold}
\usepackage[hypertexnames=false, breaklinks=true, bookmarksnumbered=true, bookmarksopen=true, colorlinks=true, linktocpage=true, citecolor=blue, urlcolor=magenta, linkcolor=magenta]{hyperref}
\makeatletter

\providecommand{\href@noop}[2]{#2}
\pdfpageheight\paperheight
\pdfpagewidth\paperwidth

\usepackage{textcomp}
\usepackage{epstopdf}

\usepackage{amsthm}
\usepackage{dsfont}
\theoremstyle{plain}
\newtheorem{theorem}{Theorem}[section]

\newtheorem{corollary}[theorem]{Corollary}

\theoremstyle{definition}

\theoremstyle{remark}
\newtheorem{remark}[theorem]{Remark}

\pdfpageheight\paperheight
\pdfpagewidth\paperwidth

\@ifundefined{textcolor}{}{%
 \definecolor{BLACK}{gray}{0}
 \definecolor{WHITE}{gray}{1}
 \definecolor{RED}{rgb}{1,0,0}
 \definecolor{GREEN}{rgb}{0,1,0}
 \definecolor{BLUE}{rgb}{0,0,1}
 \definecolor{CYAN}{cmyk}{1,0,0,0}
 \definecolor{MAGENTA}{cmyk}{0,1,0,0}
 \definecolor{YELLOW}{cmyk}{0,0,1,0}
}

\usepackage{xcolor}\usepackage{soul}
\newcommand{\bra}[1]{\ensuremath{\left\langle#1\right|}}
\newcommand{\ket}[1]{\ensuremath{\left|#1\right\rangle}}

\newcommand{\B}[0]{\mathcal B}

\newcommand{\e}[0]{\exists\,}

\newcommand{\f}[0]{\forall\,}

\renewcommand{\H}[0]{\mathcal{H}}

\newcommand{\tr}[0]{\operatorname{tr}}

\newcommand{\norm}[1]{\lVert #1\rVert}

\newcommand{\mi}{\mathrm{i}}
\newcommand{\me}{\mathrm{e}}
\newcommand{\id}{\mathbb{1}}

\definecolor{Red}{rgb}{0.70,0.13,0.13}
\definecolor{Green}{rgb}{0.13,0.55,0.13}
\definecolor{Blue}{rgb}{0.14,0.29,0.51}

\hypersetup{
  linkcolor=Red,
  citecolor=Blue,
  urlcolor=Blue
}

\makeatother
\begin{document}
\title{Synchronized Spin Trajectories under Collective Weak Measurements}
\author{Yiwen Han}
\affiliation{CESQ/ISIS (UMR 7006), CNRS and Universit\'{e} de Strasbourg, 67000 Strasbourg, France}
\author{Konghao Sun}
\affiliation{Beijing National Laboratory for Condensed Matter Physics,Institute of Physics, Chinese Academy of Sciences, Beijing 100190, China}
\author{Wei Yi}
\affiliation{Laboratory of Quantum Information, University of Science and Technology of China, Hefei 230026, China}
\affiliation{Anhui Province Key Laboratory of Quantum Network, University of Science and Technology of China, Hefei, 230026, China}
\affiliation{CAS Center For Excellence in Quantum Information and Quantum Physics, Hefei 230026, China}
\author{Johannes Schachenmayer}
\affiliation{CESQ/ISIS (UMR 7006), CNRS and Universit\'{e} de Strasbourg, 67000 Strasbourg, France}
\date{\today}
\begin{abstract}
We consider locally driven-dissipative quantum spins that undergo a continuous weak measurement of their collective spin. We show that such a system can exhibit full spin-synchronization in trajectories corresponding to the same measurement record. Mathematically, we demonstrate that this is a general consequence of the irreducibility of the quantum Markov semigroup generated by the Lindbladian, in combination with a transitive symmetry, in particular spin-permutation symmetry. When breaking the irreducibility, e.g.~in the absence of local dissipation, the system can exhibit a weaker form of partial synchronization. Given the general origin, this measurement-induced phenomenon arises in a broad class of many-body systems.     
\end{abstract}
\maketitle

\emph{Introduction---}Synchronization, the phenomenon where the motion of numerous degrees of freedom are locked together, is a paradigmatic form of dynamical order~\cite{Pikovsky2001,Kuramoto1975,Acebron2005}.
It can apply to both classical and quantum settings~\cite{Pikovsky2001,Arenas2008,Mari2013,EshaqiSani2020,Schmolke2026,Solanki2026Review}, ranging from coupled phase oscillators~\cite{Kuramoto1975,Acebron2005,Pikovsky2001,Arenas2008} to driven-dissipative many-body systems~\cite{Cabot2019,Schmolke2024,Dai2026,Solanki2026Review}.
By suppressing fluctuations in relative phases and frequencies, and by reinforcing collective signals, synchronization is relevant to sensing~\cite{Bekker2017,DuttaCooper2019,Xu2020,Xu2022,Vaidya2025} and metrology~\cite{Meiser2009,Xu2015,Roth2016}, has found applications in timekeeping, communication networks, oscillator and laser arrays, information processing, and thermal machines~\cite{Meiser2009,Roth2016,Nixon2012,Bekker2017,Jaseem2020,
Solanki2026Review,Esencan2026}.
In open quantum systems, most routes to synchronization are based on interactions between oscillators (including limit-cycle oscillators) or spins, 
involve models whose semiclassical counterparts retain the structure of Kuramoto-type problems~\cite{Heinrich2011,Walter2015,Zhu2015,Delmonte2023,Lee2013,Liu2026}, or arise from the selection of long-lived modes~\cite{Giorgi2013,Manzano2013,Buca2022}.
Within this category is a recently discovered measurement-induced synchronization~\cite{Schmolke2024}, where the measurement backaction localizes the system in decoherence-free subspaces, yielding frequency-synchronized oscillations.

In this work, we demonstrate a different type of collective measurement-induced quantum synchronization, where the dynamics of all spins matches at late times within a single trajectory corresponding to a certain measurement record. The phenomenon may be considered as a quantum counterpart of noise-assisted synchronization in an ensemble of classical uncoupled limit-cycle oscillators~\cite{Nakao2007}. However, we propose a distinct underlying mathematical mechanism: 
It derives from the fact that evolution operators form an irreducible quantum Markov semigroup (QMS) in the presence of local dissipation. As a result, time-evolved trajectories with an identical measurement record, but with different initial states approach each other at sufficiently long times.
If the  system further possesses a transitive symmetry, e.g.~spin-permutation symmetry, we then show that this leads to an identical evolution of the spin-components in the single trajectories. The absence of local dissipation leads to an incomplete local spin algebra and hence only partially synchronized dynamics. Breaking the spin-permutation symmetry, e.g.~by adding local disorder in the Hamiltonian, destroys the phenomenon. The algebraic origin of the effect is demonstrated to be general, arising e.g.~also in interacting models and for larger spins as long as the QMS irreducibility and the symmetry conditions are met. Furthermore, this synchronization is independent of the measurement basis of the ancilla used in the weak-measurement construction.

\smallskip

\emph{Minimal model---}We first consider a simple model of an ensemble of $N$ non-interacting spin-1/2 particles, with Hamiltonian $\hat{H}=\Omega \hat{S}^x+\Delta \hat{S}^z$. Here, $\hat S^\alpha=\frac12\sum_{k=1}^N \hat\sigma^\alpha_k$ are collective spin-operators with corresponding  Pauli matrices $\hat{\sigma}_k^\alpha$ ($\alpha=x,y,z$). The parameters $\Omega$ and $\Delta$ are real.
As illustrated in Fig.~\ref{Fig1}(a), we start with a random product state of the spins, subject each spin to local dissipation (spontaneous emission), and switch on a continuous weak measurement of a collective spin component, here $\hat S^z$. 
\begin{figure}[tb]
    \begin{centering}
	\includegraphics[width=\linewidth]{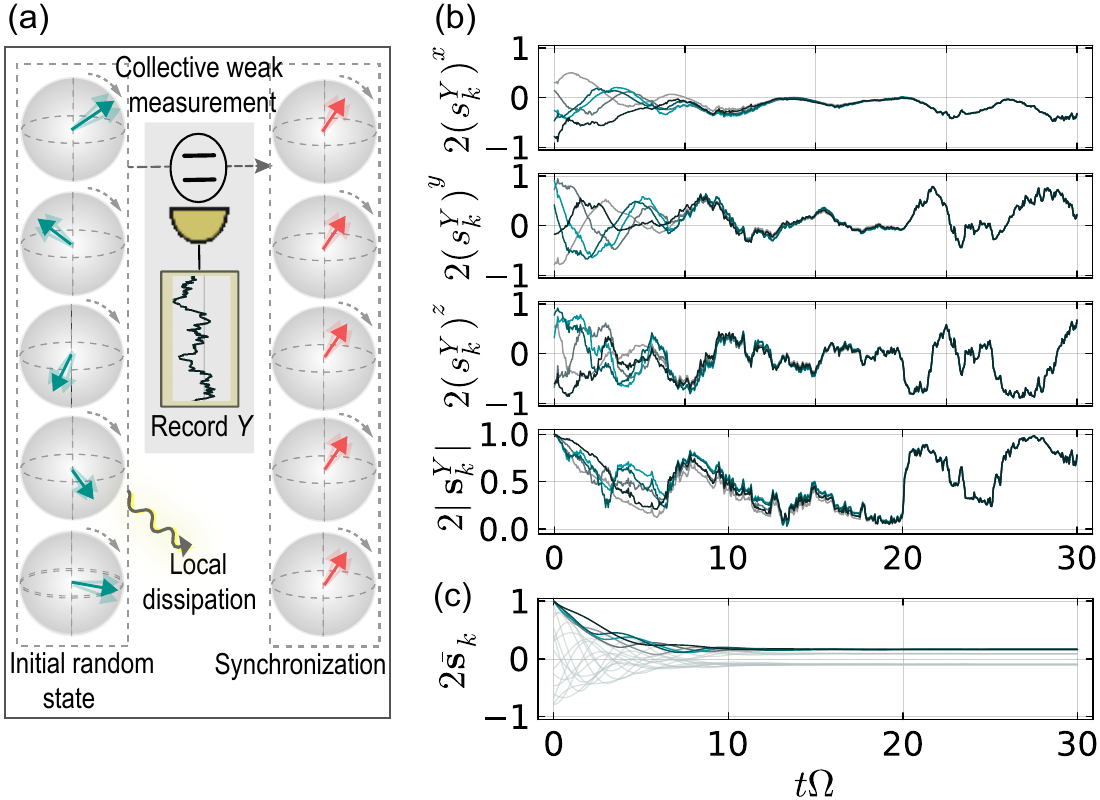}\\
    \par\end{centering}
    \caption{(a) Schematic illustration: $N$ non-interacting spin particles undergo collective continuous weak measurement, i.e.~they are jointly coupled to an ancilla, which is continuously  measured projectively. Each spin is subjected to local dissipation. Within each single trajectory all spins exhibit identical stochastic evolution at late times.
    (b) Example of the effect for $N=5$ driven non-interacting spin-$1/2$ particles that are evolving under Eq.~\eqref{eq:SME} with measurement record $Y$. Shown are results for the  single-trajectory Bloch vector components, $\bm{s}_k^Y$ for $k=1,\dots,5$ (solid lines in different colors), given a certain measurement record $Y$. Bloch-vector lengths remain finite and fluctuating.
    (c) Ensemble-averaged dynamics computed from the full density matrix, $\overline{\bm{s}}_k$. Solid green lines show Bloch-vector lengths, while thin gray curves show components $\overline{s}_k^x,\,\overline{s}_k^y,\,\overline{s}_k^z$.
    Parameters: \(N=5\), \(\Delta=0.5\Omega\), \(\gamma=0.1\Omega\), \(\Gamma=0.8\Omega\),  \(\phi=0\). }
    \label{Fig1}
\end{figure}
The measurements do not resolve individual spins, but preserve the spin-permutation symmetry, and generate a common stochastic backaction. Conditioning on a collective measurement record $Y$, we describe the dynamics by a stochastic master equation~\cite{WisemanMilburn2009,Barchielli2009,Jacobs2014} ($\hbar\equiv1$)
\begin{align}
d\hat\rho_c
=
&-\mi[\hat H,\hat\rho_c]dt+ \gamma\sum_k \mathcal{D}[\hat\sigma^-_k]\hat\rho_c dt\nonumber\\
&+\Gamma\mathcal{D}[\hat S^z]\hat\rho_c dt
+
\sqrt{\Gamma}\mathcal{H}[\me^{-\mi\phi}\hat S^z]\hat\rho_c dW_t. 
\label{eq:SME}
\end{align}
We define the super-operators $\mathcal D[\hat L]\hat\rho=\hat L\hat\rho \hat L^\dagger -\frac{1}{2}\{\hat L^\dagger \hat L,\hat\rho\}$, and $\mathcal H[\hat O]\hat\rho = \hat O\hat\rho+\hat\rho\hat O^\dagger -\tr[(\hat O+\hat O^\dagger)\hat\rho]\hat\rho$, and $dW_t$ is a Wiener increment satisfying $\mathbb E[dW_t]=0$ and $ dW_t^2= dt$. The local Lindblad operator $\hat\sigma^-_k$ governs spontaneous emission with homogeneous rate $\gamma$. 
The last term of Eq.~(\ref{eq:SME}) represents the backaction of the weak collective measurement, which features a strength $\Gamma$ and a homodyne phase $\phi$.
The corresponding homodyne record is determined by the equation
\begin{align}
dY_t=2\sqrt\Gamma\cos\phi \tr(\hat S^z\hat\rho_c)dt+dW_t.
\label{eq:measurement_record}
\end{align}
As ``trajectory'' from now on we denote a state evolution that corresponds to a single measurement record $Y$. Throughout this work, we set $\phi=0$, which amounts to measuring the informational homodyne quadrature of $\hat S^z$. Other choices do not qualitatively change our results.

Equation~\eqref{eq:SME} describes a weak measurement, where information on the collective spin variable $\hat S_z$ is inferred indirectly, by first entangling the system to an ancilla while continuously performing projective measurements on the ancilla. Generally, dynamics over an infinitesimal time-step can be described in the Kraus operator language~\cite{Nielsen2010}. There, the density matrix is updated according to the quantum operation $\hat \rho(t+dt) = \mathcal{E}[\hat \rho({t})] = \sum_k \hat K_m \hat \rho(t) \hat K_m^\dag$ with Kraus operators $\hat K_m$ fulfilling $\sum_m \hat K_m^\dag \hat K_m =\id$. Writing the operator sum in a re-normalized form, $\mathcal{E}[\hat \rho({t})] = \sum_k p_m \hat K_m \hat \rho(t) \hat K_m^\dag/p_m$ with probabilities $\sum_m p_m = 1$, the stochastic evolution for a conditioned density matrix becomes evident. Specifically, one can select and apply a random Kraus operator $\hat K_m$, followed by renormalization. The selection corresponds to the outcome of the projective ancilla measurement. Specifically (see~\cite{SM}), when choosing a qubit ancilla, a system-ancilla coupling $\propto \hat S_z \otimes \hat \sigma^y$, and measuring the ancilla in the eigenbasis of $\hat s^x$, leads to the two Kraus operators $\sqrt{2} \hat K^h_{0(1)} = \id \pm \sqrt{dt\Gamma} \hat S_z - dt \Gamma \hat S_z^2 /2$. In~\cite{SM} we show that those are indeed identical to the balanced homodyne scheme of~Eq.~\eqref{eq:SME}. However, equivalently the ancilla could be measured in the eigenbasis of $\hat \sigma_z$  (``number basis''), which then leads to $\hat K^n_{0} = \id - dt\Gamma \hat S_z^2/2$ and $\hat K^n_1 = \sqrt{dt \Gamma} \hat S_z$. Any basis choice leads to different Kraus operators that on the density matrix level give rise to identical evolution. However, on the trajectory level dynamics can strongly differ~\cite{Vovk2022, Daraban2025, Cichy2025, Rosario2025}. Importantly, the general algebraic arguments that we describe below are independent of this choice. Therefore the effect remains, regardless of the chosen unraveling.

We focus on single-trajectory dynamics generated by Eq.~(\ref{eq:SME}), where only the measurement record $Y$ is retained, while the local dissipation events remains unresolved and are hence traced over. 
For each trajectory, $\hat \rho_c^Y(t)$, we characterize the evolution of spin $k$ by its Bloch vector $\bm{s}^Y_k(t)=\tr[\hat\rho_c(t)\hat{\boldsymbol{\sigma}}_k]/2$, with $\hat{\boldsymbol{\sigma}}_k$ the vector of Pauli matrices.
We initialize the system in a product state, with random initial spin directions, $\bm{s}^Y_i(0)\neq\bm{s}^Y_j(0)$ for $i\neq j$. 
A typical single-trajectory evolution of the Bloch vectors $\mathbf{s}_k(t)$ is shown in Fig.~\ref{Fig1}(b), displaying  a gradual merging of Bloch-vector directions and lengths.
This is in contrast to the ensemble-averaged dynamics shown in  Fig.~\ref{Fig1}(c). There we show the ensemble-averaged  evolution computed from the full density matrix $\hat\rho$ via $\overline{\bm{s}}_k = \tr{[\hat \rho(t) \hat{\bm{\sigma}}_k]}/2$, and we observe that the system relaxes to a unique steady state with a given spin direction. These simple numerical results showcase that the unique and permutation-symmetric steady-state density matrix can be written as a statistical mixture of fluctuating trajectories. The weak measurements do not distinguish individual spins, therefore they lead to random uniform kicks of all spins into identical directions, which gives rise to the synchronized trajectory dynamics. Our key finding is that this effect can be understood from a general algebraic origin.

\smallskip

\emph{Algebraic origin---}Over a certain time $t>0$, the full Markovian density matrix evolution is governed by an evolution operator of the form $ \mathcal T_t=\me^{t\mathcal L}$ with \(\mathcal L\) the Lindblad generator. 
Here, \({\mathcal T}_{t\geq0}\) drives completely positive and trace-preserving evolution of the density matrix. The evolution operators form the quantum Markov semigroup (QMS), $\mathcal T_{t+s}=\mathcal T_t\mathcal T_s$.
Importantly, in our Eq.~\eqref{eq:SME}, the local dissipative terms allow to evolve the system into states on the full Hilbert space $\mathcal H$. The Lindblad operators $\hat\sigma_k^-$, together with the local coherent Hamiltonian, generate the von Neumann algebra $\mathcal B(\mathbb C^2)$ for each spin. Taking the tensor product over all spins yields the full algebra $\bigotimes_{k=1}^N \mathcal B(\mathbb C^2) = \mathcal B(\mathcal H)$. This makes the Lindbladian-generated QMS irreducible~\cite{Fagnola2026,Evans1977}. In other words, this implies that no nontrivial invariant subspace remains. A consequence of this is that there exists a unique steady state. To see this mathematically, one can show~\cite{SM} that for any two initial states \(\hat\rho_0\) and \(\hat\rho'_0\), the following inequality holds at any time $t$:
\begin{align}
\left\|\mathcal T_t(\hat\rho_0)-\mathcal T_t(\hat \rho_0^\prime)\right\|_1
\leq
C e^{-\lambda t}.
\label{eq:QMS_contraction}
\end{align}
Here, the coefficients \(C>0\) and \(\lambda>0\) are independent of the initial state, and \(\|\cdot\|_1\) denotes the trace norm. Consistent with Eq.~(\ref{eq:QMS_contraction}), Fig.~\ref{Fig1}(c) displays the evolution into the unique steady state on the density matrix level in our example model.

In addition, we show that a similar dynamic contraction also persists at the single-trajectory level (see~\cite{vanHandel2009,Amini2021,Amini2026} for other related work).
Specifically, consider the evolution of two trajectory states with an identical measurement record $Y$, but with generally different initial conditions \(\hat\rho_0\) and \(\hat\rho'_0\). In~\cite{SM}, we show that for two states with two different initial conditions, $\hat\rho_c^Y(t;\hat\rho_0)$ and $\hat\rho_c^Y(t;\hat\rho'_0)$, at any time $t$, 
\begin{align}
\norm{\hat\rho_c^Y(t;\hat\rho_0)-\hat\rho_c^{Y}(t;\hat\rho_0^\prime)}_1
\leq C^{Y} e^{-\lambda_c t}.
\label{eq:Contraction}
\end{align}
We note that $C^{Y}>0$ is trajectory dependent, and that the coefficients $C^{Y}$ and $\lambda_c$ need not take the same values as their counterparts in Eq.~\eqref{eq:QMS_contraction}.
In practical terms, Eq.~\eqref{eq:Contraction} implies that all trajectories become identical \emph{independent of the initial condition} if they exhibit the same measurement record.

Synchronization between spins in a single trajectory then follows as a consequence of the  spin-permutation symmetry. We define a swap operator exchanging spin $i$ and $j$,  \(\hat P_{ij}\), and the super-operator $\mathcal P_{ij}[\hat X]\equiv\hat P_{ij}\hat X\hat P_{ij}^{\dagger}$, such that e.g.~$\mathcal P_{ij}[\hat{\bm{\sigma}}_i]=\hat{\boldsymbol{\sigma}}_j$. Assuming that the stochastic evolution of the trajectory is governed by a permutation-symmetric equation, as for our model in Eq.~\eqref{eq:SME}, we obtain the equation $\hat\rho_c^Y \left(t;\mathcal P_{ij}[\hat\rho_0]\right)=\mathcal P_{ij}\left[\hat\rho_c^Y(t; \hat\rho_0)\right]$ for arbitrary $i$ and $j$. Defining a different initial condition by only swapping two spins $i$ and $j$, $\hat\rho_0^\prime=\mathcal P_{ij}[\hat\rho_0]$, and substituting this in Eq.~\eqref{eq:Contraction}, we arrive at an inequality for the norm of the difference of two Bloch vectors,
\begin{align}
\norm{\mathbf{s}^Y_i(t)-\mathbf{s}^Y_j(t)}\leq
\frac{\sqrt3}{2}\norm{\hat\rho_c^Y(t;\hat\rho_0)-
\mathcal P_{ij}[\hat\rho_c^{Y}(t;\hat\rho_0)]}_1\xrightarrow{t \to \infty} 0.
\label{eq:complete_sync}
\end{align}
This implies that in each trajectory any two pairs of spins evolve identically at sufficiently long times.

\emph{Synchronization dynamics---}
To analyze the dynamics we define two synchronization measures,
\begin{align}
R_d(t)\!=\!\frac{2}{N(N\!-\!1)}\sum_{i<j}\frac{\mathbf s_i(t)\cdot\mathbf s_j(t)}{|\mathbf s_i(t)||\mathbf s_j(t)|},\, R_a(t)\!=\!\frac{(\sum_{k}|\mathbf s_k(t)|)^2}{N\sum_k |\mathbf s_k(t)|^2},
\end{align}
which characterize direction and amplitude synchronization, respectively. Both quantities approach unity as the local spin vectors fully synchronize. Based on this, we define a synchronization error $D_R(t)=1-R_d(t)R_a(t)$, and  its asymptotic exponential decay as synchronization rate,
\begin{align}
\lambda_R=-\lim_{t\rightarrow\infty}
\frac{1}{2t}\ln \frac{D_R(t)}{D_R(0)}.
\end{align}
Note that the factor of $2$ is added since $D_R(t)$ is quadratic in the local spin mismatch.

Our goal is to compare the synchronization rates to the dynamics into the steady state on the density matrix level. At the level of the unconditional master equation, for the minimal model we can find a closed set of equations for the difference vector components, $(\mathbf{d_{ij}})^{x,y,z} = \overline{s}_i^{x,y,z} - \overline{s}_j^{x,y,z}$, $d \mathbf{d}_{ij}/dt = \mathbf{M} \cdot \mathbf{d}_{ij}$. 
The asymptotic relaxation dynamics of the local-spin difference is then dominated by the eigenvalue of $\mathbf{M}$ with the (negative) real part of smallest magnitude, which we denote as $\lambda_s$. Perturbatively, under the condition $(\Gamma-\gamma) \ll 2\sqrt{\Omega^2 + \Delta^2}$, relevant to the parameter regime of Fig.~\ref{Fig1}(b), one finds a linear increase of this eigenvalue with $\Gamma$, $\lambda_s \simeq \gamma + f(\Omega,\Delta) (\Gamma-\gamma)$.
For very large $\Gamma \gg \Delta, \gamma$, instead one finds $\lambda_s \simeq \gamma+{2\Omega^2}/{(\Gamma-\gamma)}$. Here, the second term is suppressed as $\Gamma$ increases, so that $\lambda_s$ approaches $\gamma$, consistent with a quantum Zeno-like dynamics.

\begin{figure}[tb]
    \centering
    \includegraphics[width=0.98\linewidth]{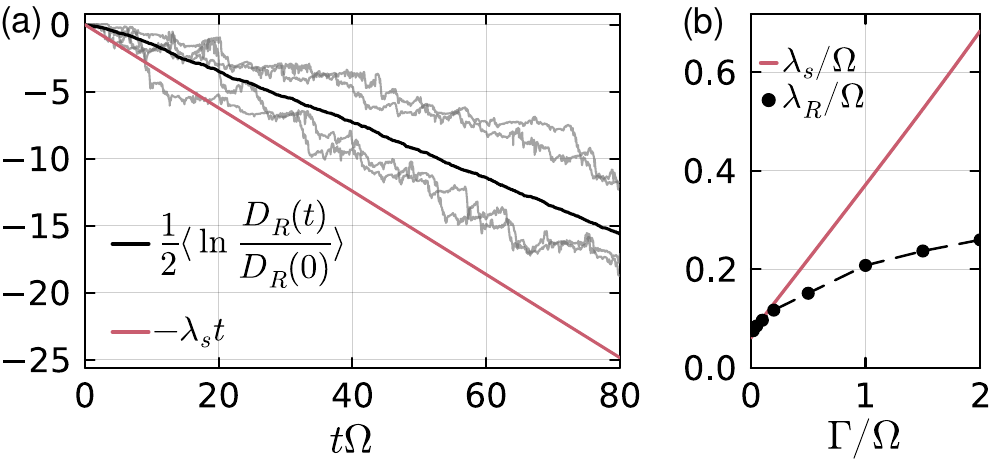}
    \caption{Dynamics of the synchronization error. (a) Gray curves: $\ln({D_R(t)}/{D_R(0)})/2$ for individual  trajectories. Black line: average over $100$ trajectories.
    At long times, both curves approach an approximately exponential decay rate $\lambda_R$. Red line: asymptotic rate from the density evolution, $\lambda_s$. 
    (b) Comparison of $\lambda_R$ and $\lambda_s$ as function of the measuring rate $\Gamma$.
    Other parameters are identical to Fig.~\ref{Fig1}(b).}
    \label{Fig2}
\end{figure}

Fig.~\ref{Fig2}(a) shows the decay of $D_R(t)$, away from the Zeno regime. Gray curves correspond to a few individual trajectories, the black line to the trajectory average. At long times, $D_R(t)$ approaches an exponential decay with corresponding synchronization rate $\lambda_R$, which we can extract with a fit. 
The red line, $-\lambda_s t$, shows the deterministic relaxation from the density matrix evolution.
Fig.~\ref{Fig2}(b) compares the two rates as a function of $\Gamma$. 
In the limit $\Gamma\to0$, the stochastic measurement term vanishes and $\lambda_R$ approaches the deterministic relaxation rate $\lambda_s$. The stochastic term scales in Eq.~\eqref{eq:SME} as $\sqrt{\Gamma}$, while the corresponding It\^o correction to the decay rate scales as $\Gamma$. Therefore, the two rates separate linearly near $\Gamma=0$. Numerically, we find that, away from the Zeno regime, the trajectory synchronization is always slower than the unconditional master equation counterpart, i.e.~$\lambda_R<\lambda_s$.

\emph{Partial synchronization---}We now analyze the fate of the synchronization if one of the two mathematical conditions is compromised. We first consider the case without local dissipation (\(\gamma=0\) in the minimal model), whereby $\mathcal{B}(\mathbb C^2)$ is no longer generated for each spin. As a result, Eq.~(\ref{eq:Contraction}) no longer holds. 
Still, we find that a partial form of synchronized trajectory dynamics remains. The Hilbert space can be resolved into sectors with different total spin, $\mathcal H=(\mathbb C^2)^{\otimes N} = \bigoplus_S \mathcal V_S\otimes\mathcal M_S$, where $\mathcal V_S$ is the irreducible spin-$S$ representation space, and $\mathcal M_S$ is the multiplicity space counting the number of equivalent copies of the spin-$S$ representation. Under this decomposition, the collective spin operators take the form $\hat S^\alpha = \bigoplus_S \hat S_S^\alpha\otimes \id_{\mathcal M_S}$, where $\id_{\mathcal M_S}$ is the identity operator on $\mathcal M_S$. In the absence of local dissipation, the Lindbladian in Eq.~(\ref{eq:SME}) acts only on the sectors $\mathcal V_S$. Defining a projector, $\hat\Pi_S$, onto the subspace $\mathcal V_S\otimes\mathcal M_S$, the single-trajectory weights $w_S^Y(t)=\tr[\hat\Pi_S \hat\rho_c^Y(t)]$ of the different sectors are independently updated by the measurement backaction.

\begin{figure}[t]
    \begin{centering}
    \includegraphics[width=0.93\linewidth]{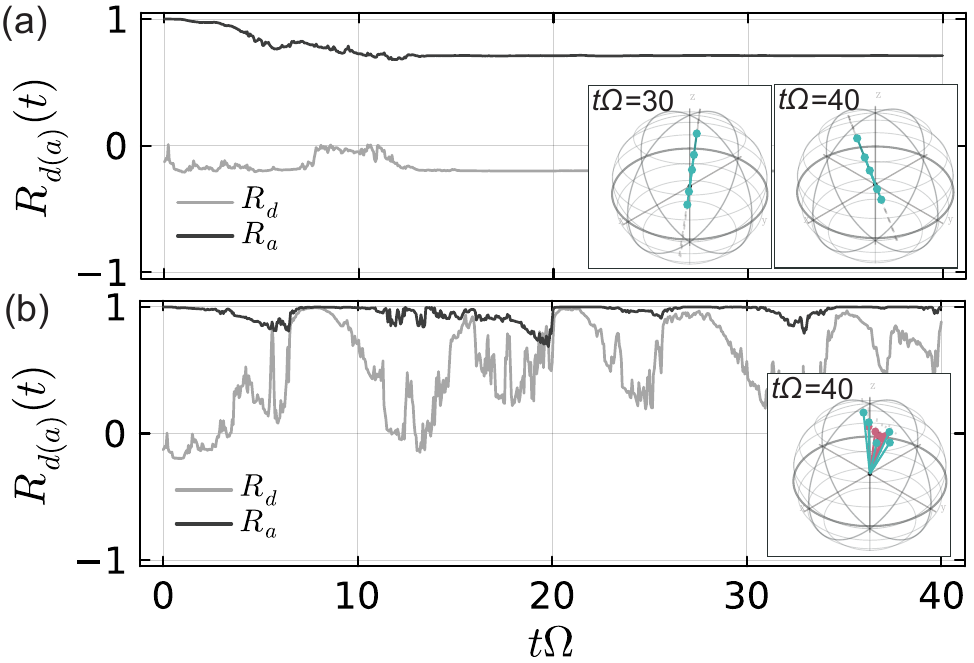}
    \par\end{centering}
\caption{Time evolution of $R_d$ and $R_a$ for: (a) dynamics without local dissipation ($\gamma=0$), and (b) in the presence of local disorder. The insets show snapshots of local Bloch-sphere distributions, small dashed lines indicate radial projections onto the Bloch-sphere surface. In (a), a partial collinear synchronization effect can be observed. In (b), synchronization breaks down. Disorder is added in detunings and Rabi frequencies, \(\Delta_k=\Delta+\delta\Delta_k\) and \(\Omega_k=\Omega(1+\delta\Omega_k)\), with \(\delta\Delta_k\) and \(\delta\Omega_k\)  Gaussian random variables with zero sample mean and disorder strengths \(W_\Delta=W_\Omega=0.3\) (green dots), \(W_\Delta=W_\Omega=0.1\) (almost aligned red dots).
 Other parameters are identical to Fig.~\ref{Fig1}(b).
}
    \label{Fig3}
\end{figure}

Through a mechanism known as dissipative freezing~\cite{Benoist2014,SanchezMunoz2019}, the system always evolves into a single total-spin sector, also for  initial states having support on multiple sectors. 
Denoting this spin sector as $S_\ast$, the QMS generated by the Lindbladian is irreducible on $\mathcal V_{S_\ast}$, and the dynamic contraction therefore survives in this subspace. As argued above, spin-permutation symmetry locks the spins onto a common stochastic direction, $\mathbf{v}^Y(t)$. Now, however, the amplitudes of each spin $k$, \(a_k^Y\), depend on the unresolved degrees of freedom in the multiplicity space \(\mathcal M_{S_*}\), and therefore will depend on $k$, which overall leads to the time-dependent spins evolving as (see~\cite{SM} for details):
\begin{align}
\mathbf{s}^Y_k(t)\xrightarrow{t \to \infty} a_k^Y\mathbf v^Y(t).
\label{eq:collinear}
\end{align}
At late times, the Bloch vectors of the spins can be parallel or antiparallel,
but their lengths do not need to be identical, only their ratios become fixed asymptotically.

An example of such a ``collinear configuration'' is  shown in Fig.~\ref{Fig3}(a), where the initial product state features \(\langle \hat{\mathbf S}^2\rangle\approx 3.1\) but the system evolves into the $S_\ast=1/2$ sector at long times. As a signature of the partial synchronization, here our synchronization measures $R_{d(a)}(t)$ evolve to a constant values different from one.
We note that, in a special case of \(S_*=N/2\), the multiplicity space is one-dimensional ($\mathcal M_{N/2}\simeq\mathbb C$), the long-time amplitudes $a_k$ become $k$-independent. In this scenario, full synchronization is recovered.

\emph{Broken symmetry---}By introducing local disorder we can fully destroy the synchronization effect. While dynamic contraction still exist, it now occurs on different time scales for different spins, since $\hat\rho_c^Y \left(t;\mathcal P_{ij}[\hat\rho_0]\right) \neq\mathcal P_{ij}\left[\hat\rho_c^Y(t;\hat\rho_0)\right]$. Hence, synchronization becomes gradually dispersive with increasing disorder, as illustrated in Fig.~\ref{Fig3}(b). There, we show that both $R_d$ and $R_a$ deviate from unity and fluctuate in time. Bloch vector misalignment becomes more pronounced when increasing the disorder strength (inset). We note that when the permutation symmetry is only partially broken, synchronization survives only among the spins with intact permutation symmetry.
Specifically, if $\hat\rho_c^Y \left(t;\mathcal P_{ij}[\hat\rho_0]\right) = \mathcal P_{ij}\left[\hat\rho_c^Y(t;\hat\rho_0)\right]$ holds for a given pair of $i$ and $j$, dynamics of the $i$-th and $j$-th spins are still fully synchronized.

\emph{Generality of the effect---}The synchronization effect studied here, arises from a purely algebraic perspective. What matters is the generation of the full operator algebra on the relevant Hilbert space, so that the full many-body Hilbert space is accessible to the dynamics. Given this, the effect can also be achieved for larger spins. For instance, synchronization of a spin-$1$ ensemble can be achieved when the local Hamiltonian and dissipation conspire to generate the full algebra $\bigotimes_{k=1}^N\mathcal B(\mathbb C^3)$.

Furthermore, while so far we have focused on a diffusive homodyne unraveling in Eq.~(\ref{eq:SME}), other choices of the ancilla measurement basis/Kraus operators also results in synchronized dynamics. 
For instance, we can also consider the number-resolved unraveling corresponding to Kraus operators $\hat K^n_{0(1)}$, one can again map the dynamics to a SDE~\cite{SM}. The binary measurement increment $dN_t\in\{0,1\}$ satisfies $\mathbb E[dN_t|\hat\rho_c] = \Gamma \tr[\hat O^\dagger\hat O\hat\rho_c] dt$. The last term in Eq.~\eqref{eq:SME} is then replaced by
\begin{align}
\mathcal J[\hat O]\hat\rho_c \left( dN_t - \Gamma \tr[\hat O^\dagger\hat O\hat\rho_c] dt\right),
\label{eq:dN}
\end{align}
where $\mathcal J[\hat O]\hat\rho = \frac{\hat O\hat\rho\hat O^\dagger}{\tr[\hat O^\dagger\hat O\hat\rho]}-\hat\rho$.
Also for this measurement scheme, the algebraic arguments remain the same and the conditional dynamics exhibits full synchronization~\cite{SM}.

\begin{figure}[htbp]
    \begin{centering}
    \includegraphics[width=0.9\linewidth]{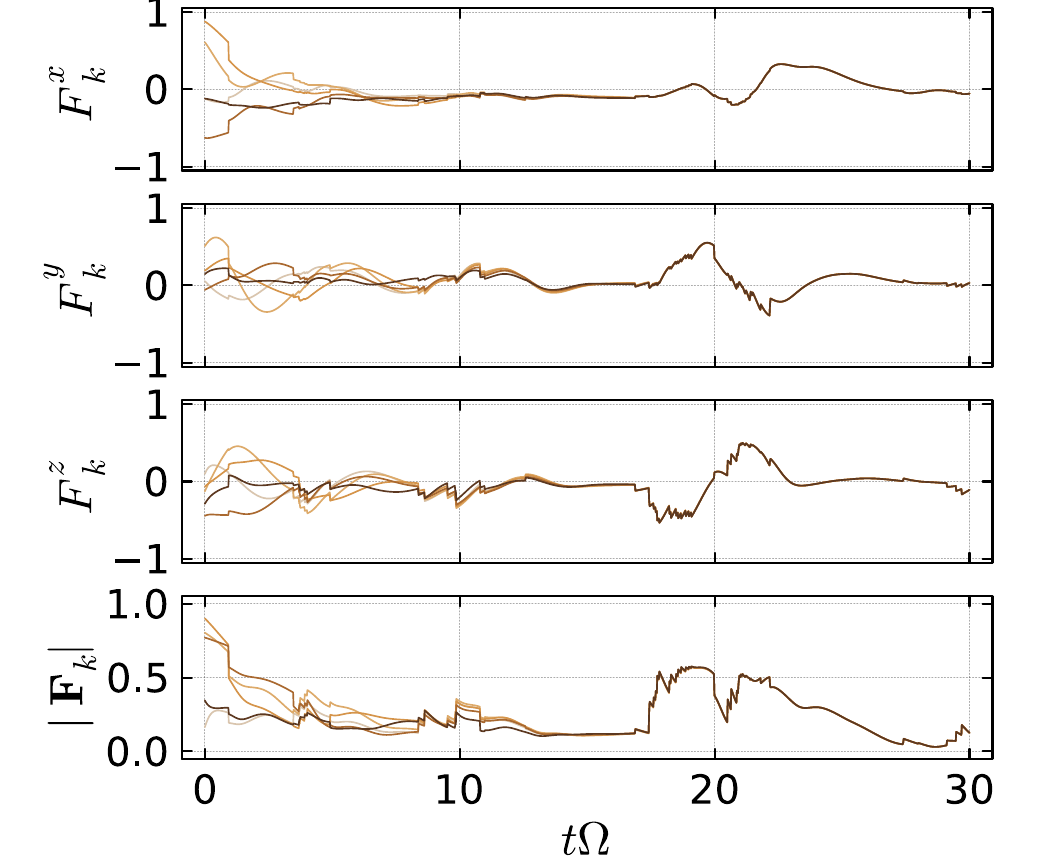}
    \par\end{centering}
    \caption{Synchronization in a dissipative ensemble of interacting spin-\(1\) particles. The system is initialized in a generic product state. Under the quantum-jump conditional dynamics, the trajectories of polarization vectors $\mathbf{F}_k(t)=\tr[\hat\rho_c(t)\hat{\boldsymbol{F}}_k]=\bigl(F_k^x(t),F_k^y(t),F_k^z(t)\bigr)$, shown as solid lines, progressively converge and eventually overlap, for the same collective measurement record. 
    Parameters: \(N=5\), \(\Delta=0.5\Omega\), \(J=0.5\Omega\), \(\gamma_{+0}=\gamma_{0-}=0.1\Omega\), \(\Gamma=0.5\Omega\). }
    \label{Fig4}
\end{figure}

Lastly, the condition of spin-permutation symmetry can be relaxed.
Specifically, we consider the general setting of an ensemble of quantum objects indexed by $k=1,\ldots,N$.
Let $G\subseteq S_N$ ($S_N$: full permutation group of the $N$ objects) denote the subgroup of permutations that are symmetries of the system (dubbed the transitive permutation subgroup). Let each element $g\in G$ be represented by a unitary operator $\hat{U}_{g}$, then the coherent Hamiltonian $\hat H$, the monitored collective observable $\hat O$, and the local quantum jump operators ${\hat L_k}$ should satisfy $\hat{U}_{g} \hat H \hat{U}_{g}^\dagger = \hat H$, $\hat{U}_{g} \hat O \hat{U}_{g}^\dagger = \hat O$, and $\hat{U}_{g} \hat L_k \hat{U}_{g}^\dagger = \hat L_{g(k)}$. 
For the complete synchronization of the  $N$ objects to occur, we require the existence of a permutation $g \in G$, such that $g(i) =j$ for arbitrary indices $i$ and $j$. 
In this light, the full spin-permutation symmetry that we introduced through the minimum model is only a special case and stronger than necessary. 

In order to illustrate the synchronization effect under all three generalizations, we consider a spin-$1$ chain with the Hamiltonian $\hat H = \sum_{k}
\left( \Omega \hat F_k^x
+ \Delta \hat F_k^z \right)
+ J \sum_{k} \left( \hat F_k^x \hat F_{k+1}^x
+ \hat F_k^y \hat F_{k+1}^y \right)$, subject to the periodic boundary condition $\hat F_{N+1}^\alpha\equiv\hat F_1^\alpha$. Here $\hat F_k^\alpha$ are the standard spin-$1$ matrices on the local basis $\{\ket{+}_k,\ket{0}_k,\ket{-}_k\}$~\cite{SM}.
The number-resolved detection in the output channel is associated with the collective jump operator $\hat O=\sum_{k=1}^{N}\hat F_k^z$ with measurement rate $\Gamma$. 
We also include local dissipative processes given by $\hat L_{k,+0}=\ket{0}_k\bra{+}$ and $\hat L_{k,0-}=\ket{-}_k\bra{0}$, with the corresponding rates $\gamma_{+0}$ and $\gamma_{0-}$. Symmetry-wise, the system and its dynamics are invariant under the cyclic translation that maps $(1,2,\cdots,N)$ to $(2,\cdots,N,1)$.
Figure~\ref{Fig4} demonstrates the single-trajectory synchronization for this generalized model.

\emph{Conclusion \& Outlook---}We identified a measurement-induced synchronization effect at the single-trajectory level in collectively monitored quantum many-body systems. Physically, the collective measurement backaction is homogeneous, which implies that the steady-state density matrix can be unraveled into trajectories that fluctuate in a synchronized fashion. We have mathematically demonstrated the generality of the effect by reducing it to two conditions: i) a Lindbladian-generated irreducible QMS that generates dynamics on the full Hilbert space (which can e.g.~require the presence of local Lindblad channels); ii) transitive symmetry, e.g.~permutation symmetry. The effect can thus be engineered in a general class of many-body systems. We analyzed its robustness and the synchronization dynamics in several examples. It would be interesting to observe or utilize the effect experimentally. Here, the fact that the phenomenon is not only present in homodyne, but also in number-unravelings may lower requirements on collecting post-selection statistics. Furthermore, it would be interesting to explore, from an algebraic perspective, other forms of dynamical order at the single-trajectory level, or analyze ways to stabilize dynamical order using feedback. Such questions may be relevant to experimental platforms that offer a combination of collective and local monitoring, such as in Tavis-Cummings model realizations within superconducting circuit-QED \cite{Fink2009,Lalumiere2010,Filipp2009}, or with optically trapped cold atoms in cavities~\cite{picot2026extended, ye2026controlling, grinkemeyer2025error,hartung2024a}. In particular the latter have already demonstrated weak measurement implementations of the collective spin variables via dispersive atom-cavity couplings as in our minimal model~\cite{cox2016deterministic}.

{
\renewcommand{\addcontentsline}[3]{}
\begin{acknowledgments}
\emph{Acknowledgments---}We thank Yuto Ashida, Haowei Li, Arghavan Safavi-Naini for helpful discussions, and Junhui Qin for valuable mathematical input. This work was supported by the ERC Consolidator project MATHLOCCA (Grant nr.~101170485).  WY acknowledges support from the National Natural Science Foundation of China (Grants No. 12374479).
\end{acknowledgments}
}

{
\renewcommand{\addcontentsline}[3]{}
\bibliographystyle{apsrev4-2}
\bibliography{main}
}
\clearpage
\onecolumngrid
\renewcommand{\thefigure}{S\arabic{figure}}
\setcounter{figure}{0}
\pagenumbering{arabic}
\setcounter{page}{1}
\renewcommand{\theequation}{S.\arabic{equation}}
\setcounter{equation}{0}
\setcounter{section}{0}
\setcounter{tocdepth}{2}
\renewcommand{\thesection}{SM\arabic{section}}

\begin{center}
    \textbf{\Large{\textit{Supplemental Material for} \\ \smallskip
        ``Synchronized Spin Trajectories under Collective Weak Measurements''}}\\
    \hfill \break
    \smallskip
\end{center}

In this Supplemental Material, we provide additional details on the results stated in the main text. 
In SM1, we provide some mathematical background of QMS. In SM2, we discuss the weak-measurement construction by Kraus operator language, and connect it to the main text. 
In SM3, we give a proof of the stochastic dynamical contraction, and show that the synchronization effect is independent of the measurement basis. 
In SM4, we take some examples and show the measurement-induced  phenomenon arises in a broad class of many-body systems. 
Finally, in SM5, we provide further details on the partial collinear synchronization.
\tableofcontents
\vspace{1em}

\section{Quantum Markov semigroups: irreducibility and contraction}\label{section:SM1}
In this section, we provide supplementary materials for quantum Markov semigroups (QMS).
Briefly, a QMS describes the dynamic of a Markovian quantum open system (for example, a system determined by Lindblad master equations).
Here we consider a finite-dimensional Hilbert space $\mathcal{H}$.
Let $\mathcal{B}(\mathcal{H})$ be the matrix algebra over $\mathcal{H}$,
which is viewed as a $C^*$-algebra consisting of bounded linear operators on $\mathcal{H}$ equipped with the operator norm.

A QMS collects a continuous family $\{\mathcal{T}_t\}_{t\ge 0}$ of bounded maps on $\mathcal{B}(\mathcal{H})$ such that:
\begin{enumerate}
    \item $\mathcal{T}_0=\operatorname{id}$ the identity map on $\mathcal{B}(\mathcal{H})$.
    \item $\mathcal{T}_t$ is unital, i.e. $\mathcal{T}_t(1)=1$ for any $t$.
    \item $\mathcal{T}_{t+s}=\mathcal{T}_t\circ \mathcal{T}_s$ for any $t,s\ge 0$.
    \item $\mathcal{T}_t$ is completely positive.
\end{enumerate}

In general, a QMS does not necessarily come from a Lindblad master equation.
However, when restricting in the finite dimensional case, such a QMS always has the following presentation in Lindbladian form:
\begin{align}
\mathcal{T}_t(x)=e^{\mathcal{L} t}(x),
\mathcal{L}(x)=\hat G^{\dagger}x+x\hat G+\sum_k \hat L_k^{\dagger} x\hat L_k,\quad x\in\mathcal{B}(\H)
\end{align}
For example, if we consider the QMS coming from a Lindbladian master equation, then the operator $\hat G$ is a linear combination of the Hamiltonian $\hat H$ and $\hat L_k^{\dagger}\hat L_k$.

The following theorem states the relation between the algebraic structure and the dynamic of the QMS.

\begin{theorem}
\label{positivity_improving}
The following conditions on a QMS are equivalent,
which we call by being irreducible:
	\begin{enumerate}
    \renewcommand{\labelenumi}{(\arabic{enumi})}
		\item $\hat G$ and $\hat L_k$ generate the algebra $\mathcal{B}(\mathcal{H})$.
		\item $\mathcal{T}$ is positivity improving,
        i.e. if an element $x\in \mathcal{B}(\mathcal{H})$ satisfies $x\ge0$ and $x\neq 0$,
        then we have $\mathcal{T}_t(x)>0$.
	\end{enumerate}
\end{theorem}
\begin{remark}
    The assumption of finite dimensionality is essential to produce the equivalence, as discussed in \cite{Fagnola2026}.
    They are equivalent to another notion of the irreducibility of a QMS in \cite{Fagnola2026},
    whence the notation;
    while the precise definition of the irreducibility is not so important for us.
\end{remark}

The importance of irreducibility is that an irreducible QMS has strong restrictions on its dynamics, which we present below.
We start by a general input.
\begin{theorem}[Birkhoff contraction theorem~\cite{Bushell1973}]\label{Birkhoff}
	Let $C$ be a closed, pointed, convex cone in a real Banach space, and let
	$T:C\to C$ be a positive linear operator. For $x,y\in C^\circ$ in the interior of $C$, we define
	Hilbert's projective metric~\cite{ReebKastoryanoWolf2011} by
	\begin{align}
	d_H(x,y)
	=
	\log \frac{M(x/y)}{m(x/y)},
	\end{align}
	where
	$ M(x/y) = \inf\{\zeta>0 : x\leq \zeta y\}$ and $m(x/y) = sup\{\eta>0 : \eta y\leq x\}$.
	The projective diameter of $T$ is defined as
	$\Delta(T)=\sup_{x,y\in C^\circ} d_H(Tx,Ty)$.
	Then for all $x,y\in C^\circ$,
	we have
	\begin{align}
	d_H(Tx,Ty)
	\leq
	\tanh\!\left(\frac{\Delta(T)}{4}\right)d_H(x,y).
	\end{align}
\end{theorem}

The following states an exponential decay of contraction for an irreducible QMS. 

\begin{corollary}
	Let $\mathcal{T}$ be an irreducible QMS on a finite dimensional Hilbert space $\mathcal{H}$.
	Choose two density matrices $x,y\in \mathcal{B}(\mathcal{H})$.
    Let $\|x\|_1$ be the trace norm of $x$,  i.e. the sum of all singular values of $x$.
	Then there exist constants $c,C>0$ such that
	\begin{align}
	\|\mathcal{T}_t x-\mathcal{T}_t y\|_1\le C e^{-ct},\f t>0.
	\end{align}
\end{corollary}
\begin{proof}
	Fix some $t_0$.
	Since $\mathcal{T}$ is positivity improving,
	$K:=\mathcal{T}_{t_0}(\Sigma)$ is compact for a compact $\Sigma$.
	Choosing $\Sigma$ the unit ball in $\B(H)$,
	we obtain $\Delta(\mathcal{T}_{t_0})=\sup_{x,y\in K}d_H(x,y)<\infty$.
	Then applying the Birkhoff contraction theorem,
	we get 
	\begin{align}
	d_H(\mathcal{T}_{t_0}x,\mathcal{T}_{t_0}y)\le q d_H(x,y)
	\end{align}
	for $q<1$. This also implies that 
	\begin{align}
	d_H(\mathcal{T}_{nt_0}x,\mathcal{T}_{nt_0}y)\le q^n d_H(x,y).
	\end{align}
	Now for any $t>0$,
	we can write $t=nt_0+s$, where $0<s<t_0$,
	then 
	\begin{align}
	d_H(\mathcal{T}_tx,\mathcal{T}_ty)\le q^n d_H(\mathcal{T}_sx,\mathcal{T}_sy)\le q^n d_H(x,y)\le C e^{-ct},
	\end{align}
	where $c=-\frac{\log q}{t_0}$ and $C=q^{-1} d_H(x,y)$.
	However, if $A$ and $B$ are two density matrices then $\|A-B\|_1\le \frac{d_H(A,B)}{2}$, hence we also deduce that 
          \begin{align}
    	\|\mathcal{T}_t x-\mathcal{T}_t y\|_1\le C e^{-ct},\f t>0.
         \end{align}
    \qedhere
\end{proof}

\section{Master equations unravelings with Kraus operators}
Generally, within the weak-measurement framework, the dynamics over an infinitesimal time step can be described using the language of Kraus operators~\cite{Nielsen2010}. 
In this section, we will discuss the synchronization effect under different measurements/unravelings.

\subsection{Kraus operators}
At the level of a short Born--Markov evolution step, the operator-sum representation provides a general description of the dynamics,
\begin{align}
    \mathcal E(\hat\rho)=\sum_m\hat K_m\hat\rho\hat K_m^\dagger,
\end{align}
with Kraus operators $\hat K_m$ that fulfill $\sum_{m}\hat K_{m}^{\dagger}\hat K_m=\id$.
For a given Kraus decomposition, the action of the channel can be interpreted as an ensemble of conditional state updates. Quantum trajectories follow from this, since one can write
 \begin{align}
\mathcal E(\hat{\rho})=\sum_{m}p_{m}\frac{\hat K_m\hat\rho\hat K_m^{\dagger}}{\tr(\hat K_m\hat\rho\hat K_m^{\dagger})}
\end{align}
with $\sum_m p_m=1$. So with probability $p_{m}$ we can randomly replace the state with the normalized one after Kraus operator application. This can be done on the state-vector level.

The Kraus representation of a quantum channel is not unique. Let $\{\hat K_m\}$ be one Kraus representation of $\mathcal E$. Another set $\{\tilde{\hat{K}}_m\}$ represents the same channel
\begin{align}
    \mathcal E(\hat\rho)=\sum_m\hat K_m\hat\rho\hat K_m^\dagger=\sum_m\tilde{\hat{K}}_m\hat\rho\tilde{\hat{K}}_m^\dagger,
\end{align}
 as long as
\begin{align}
    \tilde{\hat{K}}_{m'}=\sum_{m=1}^{q} v_{m'm}\hat K_{m} \quad \text{or} \quad \begin{bmatrix}
        \tilde{\hat{K}}_1\\
        \vdots\\
        \tilde{\hat{K}}_{q'}
    \end{bmatrix}=
    \mathbf V \begin{bmatrix}
        \hat K_1\\
        \vdots\\
        \hat K_q
    \end{bmatrix},
\end{align}
with $\mathbf V=(v_{m'm})$ an semi-unitary matrix. 
Physically, this freedom corresponds to the freedom to choose the measurement basis of the environment after it has interacted and become entangled with the system. Different choices of this environmental measurement basis lead to different conditional state updates and hence to different quantum trajectories.

We now specialize this general construction to the model considered in the main text. Its unconditional dynamics is governed by
\begin{align}
    \frac{d\hat\rho}{dt} = -i[\hat H,\hat\rho] + \gamma \sum_{k=1}^{N} \mathcal D[\hat L_k]\hat\rho + \Gamma \mathcal D[\hat O]\hat\rho.
\end{align}
The operators ${\hat L_k}$ describe unmonitored dissipative channels, whereas $\hat O$ denotes the monitored channel. 
The freedom in choosing an unraveling concerns the measurement performed on  the environmental output associated with $\hat O$. The Hamiltonian and the unmonitored channels remain the same for all unravelings considered below.

Over an infinitesimal interval $\Delta t$, a reference Kraus representation of the monitored channel is
\begin{align}
\mathcal E_{\mathrm{coll}}^{\Delta t}(\hat\rho) = \hat K_0\hat\rho\hat K_0^\dagger + \hat K_1\hat\rho\hat K_1^\dagger
+o(\Delta t^2),
\end{align}
with
\begin{align}
    \hat K_0 &=\id - \frac{\Gamma \Delta t}{2}\hat O^\dagger\hat O,\label{eq:K0}\\
    \hat K_1 &= \sqrt{\Gamma \Delta t}\hat O.\label{eq:K1}
\end{align}
Quantum trajectory unraveling according to those Kraus operators gives the usual ``number unraveling'' (Quantum Monte-Carlo wave function algorithm).

In general we may change the Kraus operators with a $2\times 2$ unitary parametrized by two angles
\begin{align}
\begin{bmatrix}
    \tilde{\hat K}_0\\
    \tilde{\hat K}_1
\end{bmatrix}
=
\begin{bmatrix}
    \cos\theta
    &
    e^{-i\phi}\sin\theta
    \\
    -e^{i\phi}\sin\theta
    &
    \cos\theta
\end{bmatrix}
\begin{bmatrix}
    \hat K_0\\
    \hat K_1
\end{bmatrix}.
    \label{eq:Klaus_Operator}
\end{align}
Here, $\theta$ is a mixing angle between number measurement and homodyne unraveling, and $\phi$ is the homodyne phase.

Taking the example choice of homodyne-unraveling Kraus operators with mixing ($\theta=\pi/4, \phi=0$; an irrelevant global phase of $\tilde{\hat K}_1$ is omitted):
\begin{align}
    \tilde{\hat K}_0
    &= \frac{1}{\sqrt{2}}
    \left[\id + \sqrt{\Gamma \Delta t}\hat O -\frac{\Gamma \Delta t}{2}\hat O^\dagger\hat O \right],\label{eq:tildeK0}
    \\
   \tilde{\hat K}_1
    &=
    \frac{1}{\sqrt{2}}
    \left[\id -\sqrt{\Gamma \Delta t}\hat O - \frac{\Gamma dt}{2} \hat O^\dagger\hat O\right].
    \label{eq:tildeK1}
\end{align}
Keeping the linear terms is important, as they give rise to linear terms in the Kraus channel.

\subsection{Qubit ancilla}
Kraus operators follow from a process where the system is entangled with an ancilla system, followed by measurement of that ancilla, or equivalently, by just changing the basis of the partial trace operation after the interaction.

Specifically, the operation leading to the Kraus operator is:

\begin{enumerate}
\item Consider a separable initial system-bath state (Born approximation),
\begin{align}
\hat{\xi}_0 = \hat{\rho}_0 \otimes |0\rangle\langle 0|.
\end{align}
Here, the left part is the system, and the right part the ancilla. We choose a specific initial ancilla state, $\ket{0}$, which is in principle arbitrary. 
Let's call it the reset state. In practice, in a continuously measured system we need to reset to this state after every measurement.

\item We turn on a coupling that entangles the system with the ancilla via a unitary $\hat{U}$, leading to the state
\begin{align}
\hat{\xi}_1
=\hat{U}\hat{\xi}_0\hat{U}^\dagger
=\hat{U}\left(\hat{\rho}\otimes |0\rangle\langle0|\right)\hat{U}^\dagger.
\end{align}

\item We trace over the bath to obtain the reduced system state.
\begin{align}
\hat{\rho}_{\mathrm{red}}
=\tr_{\mathrm{anc}}(\hat{\xi}_1)
=\sum_b\langle b|\hat{\xi}_1|b\rangle.
\end{align}
Here, $\ket{b}$ are arbitrary basis states of the ancilla. 
This is the mechanism that we can simulate by random projective measurements onto the bath basis $\{\ket{b}\}$. Throwing away the information coming out of this measurement is the Markov approximation.
\end{enumerate}

From this operation, the Kraus operators simply follow from the bath-matrix elements corresponding to this process,
\begin{align}
\hat{\rho}_{\mathrm{red}}
=\sum_b \langle b| \hat U\left(\hat\rho\otimes|0\rangle\langle0|\right)\hat U^\dagger|b\rangle
=\sum_b \hat K_b \hat\rho \hat K_b^\dagger,
\qquad \hat K_b \equiv \langle b|\hat U|0\rangle .
\end{align}

Here, as ancilla we focus on a qubit, and define the Pauli bath operators $\hat B = \hat\sigma^{x,y,z}$, with the property $\hat B^2 = \id$. We will always use the reset state $\ket{0}$ of the qubit, i.e. the $-1$ eigenstate of $\hat\sigma^{z}$.

Let's first consider an entangling operation with a hermitian system operator $\hat{A}$, $\hat A^\dagger = \hat A$. We define a Hamiltonian coupling with coupling strength $g$ as
\begin{align}
\hat H_g=g \hat A \otimes \hat B,
\end{align}
and consider a unitary evolution over a dimensionless timestep $\tau = \Delta t\,g$, the unitary evolution is
\begin{align}
\hat U=e^{-i\tau\hat A\otimes\hat B}&=\cos(\tau\hat A)\otimes\id-i\sin(\tau\hat A)\otimes\hat B=\id\otimes\id-i\tau\hat A\otimes\hat B-\frac{\tau^2}{2}\left(\hat A^2\otimes\id\right)+\mathcal{O}(\tau^3).
\end{align}
Depending on the entangling time, for an ancilla measurement basis $\{\ket{b_0},\ket{b_1}\}$, the Kraus operators are
\begin{align}
\hat K_0
&=\bra{b_0}\hat U\ket{0}=\cos(\tau\hat A)\langle b_0|0\rangle -i\sin(\tau\hat A)\bra{b_0}\hat B\ket{0}
=\left(\id-\frac{\tau^2}{2}\hat A^2\right)\langle b_0|0\rangle -i\tau\hat A\bra{b_0}\hat B\ket{0}+
\mathcal{O}(\tau^3),\\
\hat K_1
&=\bra{b_1}\hat U\ket{0}=\cos(\tau\hat A)\langle b_1|0\rangle-i\sin(\tau\hat A)
\bra{b_1}\hat B\ket{0}=\left(\id-\frac{\tau^2}{2}\hat A^2\right)\langle b_1|0\rangle-i\tau\hat A\bra{b_1}\hat B\ket{0}+\mathcal{O}(\tau^3).
\end{align}
For a continuous-measurement limit it is necessary to retain terms through second order in $\tau$, because later we identify $\tau^2$ with a real time increment.

\subsection{Homodyne unraveling and Wiener statistics}
Take $\hat B=\hat\sigma^y$, and choose the basis  $\ket{b_0}=\ket{+}$, $\ket{b_1}=\ket{-}$ with $\ket{b_\pm}=(\ket{0}\pm \ket{1})/\sqrt{2}$. Then, using $\bra{1}\hat\sigma^y\ket{0}=i$, the two Kraus operators become
\begin{align}
\hat K_0^h
&=\frac{1}{\sqrt2}\cos(\tau\hat A)+\frac{1}{\sqrt2}\sin(\tau\hat A)
\simeq\frac{1}{\sqrt2}\id+\frac{\tau}{\sqrt2}\hat A-\frac{\tau^2}{2\sqrt2}\hat A^2+\mathcal{O}(\tau^3) \\
\hat K_1^h
&=\frac{1}{\sqrt2}\cos(\tau\hat A)-\frac{1}{\sqrt2}\sin(\tau\hat A)
\simeq\frac{1}{\sqrt2}\id-\frac{\tau}{\sqrt2}\hat A-\frac{\tau^2}{2\sqrt2}\hat A^2+\mathcal{O}(\tau^3).
\end{align}
Therefore, choosing an entangling time $\tau=\sqrt{\Gamma\Delta t}$, the Kraus operators and the resulting unraveling is equivalent to the homodyne unraveling from Eqs.~\eqref{eq:tildeK0}--\eqref{eq:tildeK1} with $\hat O=\hat A$.

We now derive the corresponding stochastic differential equation.
Let's define a general ancilla measurement basis as
\begin{align}
\ket{b_0}&=\cos(\theta)\ket{0}+e^{i\phi}\sin(\theta)\ket{1},\\
\ket{b_1}&=-e^{-i\phi}\sin(\theta)\ket{0}+\cos(\theta)\ket{1}.
\end{align}
For the short-time expansion of the Kraus operators, this leads to the notation
\begin{align}
\hat K_0^h&=\cos(\theta)\left(\id-\frac{\tau^2}{2}\hat A^2\right)+\tau\sin(\theta)e^{-i\phi}\hat A+\mathcal{O}(\tau^3),\\
\hat K_1^h&=-\sin(\theta)e^{i\phi}\left(\id-\frac{\tau^2}{2}\hat A^2\right)+\tau\cos(\theta)\hat A+\mathcal{O}(\tau^3).
\end{align}
Now for a general system state we can compute the sample probabilities for the two possible outcomes:
\begin{align}
p_0=\langle\hat K_0^{h\dagger}\hat K_0^h\rangle
&=\cos^2(\theta)\left(1-\tau^2\langle\hat A^2\rangle\right)+2\cos(\theta)\sin(\theta)\cos(\phi)\tau\langle\hat A\rangle
+\tau^2\sin^2(\theta)\langle\hat A^2\rangle+
\mathcal{O}(\tau^3),\\
&=\cos^2(\theta)+\sin(2\theta)\cos(\phi)\tau\langle\hat A\rangle-\tau^2\cos(2\theta)\langle\hat A^2\rangle+\mathcal{O}(\tau^3).\\
p_1=\langle\hat K_1^{h\dagger}\hat K_1^h\rangle
&=\sin^2(\theta)-\sin(2\theta)\cos(\phi)\tau\langle\hat A\rangle+\tau^2\cos(2\theta)\langle\hat A^2\rangle+\mathcal{O}(\tau^3).
\end{align}
We now want to turn this random process into a stochastic differential equation. Therefore, we define a random variable $r=\pm 1$ corresponding to the measurement outcomes of selecting $\hat K_0^h$ and $\hat K_1^h$, respectively. The expectation value of this variable is
\begin{align}
\bar r=p_0-p_1=\cos(2\theta)+2\tau\langle\hat A\rangle\sin(2\theta)\cos(\phi)+\mathcal{O}(\tau^2).
\end{align}
Here, we will not need the second order terms as it will become evident. The second moment and variance are
\begin{align}
\overline{r^2}&=p_0+p_1=1+\mathcal{O}(\tau^3).\\
\Delta r&=\overline{r^2}-\bar r^2=1-\cos^2(2\theta)+\mathcal{O}(\tau^2)=\sin^2(2\theta)+\mathcal{O}(\tau^2).
\end{align}
From this, let's define a random time-dependent variable whose mean value drifts linearly in time $t=m\,dt$ and whose variance grows as $\sim t$. Such a variable can be mapped onto a Wiener process. 
We therefore introduce the physical time step $dt$ through the relation $\tau^2 = \Gamma\,dt$,
\begin{align}
y= \sqrt{dt} \frac{r-\cos(2\theta)}{\sin(2\theta)}.
\label{eq:y}
\end{align}
The two possible values are
\begin{align}
r=1
&\quad\Longrightarrow\quad
y=\sqrt{dt}\frac{1-\cos(2\theta)}{\sin(2\theta)}
=\sqrt{dt}\frac{2\sin^2(\theta)}{2\sin(\theta)\cos(\theta)}
=\sqrt{dt}\tan(\theta),\nonumber\\[1ex]
r=-1
&\quad\Longrightarrow\quad y=-\sqrt{dt}\cot(\theta).
\end{align}
Then we get
\begin{align}
\bar y&=2\sqrt{\Gamma}\,dt\,\langle\hat A\rangle\cos(\phi)+\mathcal{O}(dt^{3/2}),\\
\overline{y^2}&=dt\frac{\overline{r^2}-2\bar r\cos(2\theta)+\cos^2(2\theta)}{\sin^2(2\theta)}=dt+\mathcal{O}(dt^{3/2}).
\end{align}
Note that in the ideal homodyne, $\cos(2\theta)=0$ and $\sin(2\theta)=1$. In the equations for $y$, the angle $\theta$ does not show up, however, for the pure number unraveling $y$ is ill-defined since then $\sin(2\theta)\to 0$.
From the random variable $y$ we can now define a Wiener increment (with zero mean) as
\begin{align}
dW\equiv y-\bar y =y-2dt\sqrt{\Gamma}\langle\hat A\rangle\cos(\phi).
\label{eq:dW}
\end{align}
This random variable now fulfills in linear order in $dt$
\begin{align}
\overline{dW}&=0,\\
\overline{dW^2}&=dt+\mathcal{O}(dt^{3/2}),
\end{align}
it's a proper Wiener increment. The measurement outcome can be expressed as function of $dW$. From Eq.~\eqref{eq:y} and
~\eqref{eq:dW} it follows that
\begin{align}
r=y\frac{\sin(2\theta)}{\sqrt{dt}}+\cos(2\theta)
=\left[\frac{dW}{\sqrt{dt}}+2\sqrt{dt\Gamma}\langle\hat A\rangle\cos(\phi)\right]\sin(2\theta)+\cos(2\theta).
\end{align}
This equation relates the Wiener process to the random outcomes of $\hat K_0^h$ or $\hat K_1^h$. For the ideal homodyne case only the first term is relevant.

\subsection{Number measurement unraveling and counting statistics}
The number unraveling is obtained by choosing $\theta=0$, such that $\ket{b_0}=\ket{0}$, $\ket{b_1}=\ket{1}$. The two outcomes distinguish whether the ancilla remains in its vacuum state or acquires one excitation. 
The general Kraus operators then reduce to
\begin{align}
\hat K_0^n&=\cos(\tau\hat A)\simeq\id-\frac{\tau^2}{2}\hat A^2+\mathcal{O}(\tau^3),\\
\hat K_1^n&=\sin(\tau\hat A)\simeq\tau\hat A+\mathcal{O}(\tau^3).
\end{align}
Now the Kraus are equivalent to those of the number unraveling in Eqs.~\eqref{eq:K0} and \eqref{eq:K1}, again
choosing entangling time $\tau=\sqrt{\Gamma\Delta t}$ with identifying $\hat O=\hat A$.

The corresponding measurement probabilities are
\begin{align}
p_0&=1-\Gamma\langle\hat A^2\rangle dt+\mathcal{O}(dt^2),\\
p_1&=\Gamma\langle\hat A^2\rangle dt+\mathcal{O}(dt^2).
\end{align}

Introduce the counting increment $dN_t\in\{0,1\}$, with $dN_t=1$ corresponding to a detection event and $dN_t=0$ to no detection. Its conditional mean is
\begin{align}
\overline{dN}=\nu_t dt,\quad \nu_t=\Gamma\langle\hat A^2\rangle.
\end{align}
 $d N_t$ is the increment of what is mathematically called a point process, it is binary and satisfies $(dN_t)^2=dN_t$. $\nu_t$ is the conditional counting intensity, which depends on the state. 
Only when the intensity is constant, i.e. $\nu_t=\nu$, does the process reduce to an ordinary Poisson process, for which the accumulated number of counts obeys
\begin{align}
\mathbb{P}(N_t=n)=e^{-\nu t}\frac{(\nu t)^n}{n!}.
\end{align}

\subsection{Stochastic differential equations}
For $\theta=0$ (or $\theta=\pi/2$), the measurement corresponds to the number unraveling, and the record is described by a counting process $dN_t$. For any fixed $0<\theta<\pi/2$, the measurement admits a diffusive limit and is described by a Wiener increment $dW_t$, wherein $\theta=\pi/4$ giving the standard balanced homodyne case. After centering and rescaling the measurement record, the dependence on $\theta$ drops out of the resulting stochastic differential equation.

Accordingly, the corresponding normalized stochastic master equations are
\begin{align}
d\hat\rho_c
&=\mathcal{L}\hat\rho_c dt+\mathcal{J}[\hat O]\hat\rho_c \left(dN_t-\Gamma\langle \hat O^\dagger\hat O\rangle dt\right),
\qquad \text{(number unraveling)},\\
d\hat\rho_c
&=\mathcal{L}\hat\rho_c dt+\sqrt{\Gamma} \mathcal{H}[e^{-i\phi}\hat O]\hat\rho_c dW_t,
\qquad \text{(homodyne unraveling)}.
\end{align}
where
\begin{align}
\mathcal{J}[\hat O]\hat\rho&=\frac{\hat O\hat\rho\hat O^\dagger}{\tr\left[\hat O^\dagger\hat O\hat\rho\right]}-\hat\rho,\\
\mathcal{H}[\hat O]\hat\rho&=\hat O\hat\rho+\hat\rho\hat O^\dagger-\tr\left[(\hat O+\hat O^\dagger)\hat\rho\right]\hat\rho .
\end{align}

\section{Stochastic contraction}
In this section, we establish the single-trajectory contraction for the two monitoring schemes introduced above, namely the number and diffusive homodyne measurements unravelings.

The underlying logic is that an irreducible QMS gives contraction for the averaged unconditional dynamics, as discussed in \ref{section:SM1}. 
The statistics of the measurement noise then allow this contractive behavior to be inherited by the conditioned dynamics along single trajectories, yielding stochastic contraction.

Let $\mathcal{L}$ induce an irreducible QMS.
Then by Theorem ~\ref{positivity_improving}, $ e^{\mathcal{L}\tau}(\hat X)>0$, where $\hat X$ is the density matrix, $\tau$ is a brief period of evolution time, hence we have
\begin{align}
    a \tr(\hat X)\hat I
    \leq
    e^{\mathcal{L}\tau}(\hat X)
    \leq
    b \tr(\hat X)\hat I,
    \qquad
    \f \hat X\geq0,
    \label{eq:bound0}
\end{align}
for some $0<a<b<\infty$ independent of $\hat X$.
For a fixed measurement record $\omega$ over $[t,t+\tau]$, let $\Phi_\omega$ denote the corresponding unnormalized conditional evolution, including the Hamiltonian and unmonitored local dissipative evolution generated by $\mathcal L_0$. The map $\Phi_\omega$ is completely positive. For any density matrix $\hat X$ such that $\tr[\Phi_\omega(\hat X)]>0$.
The normalized conditional map is
\begin{align}
    \widetilde{\Phi}_{\omega}(\hat X)=\frac{\Phi_{\omega}(\hat X)}{\tr[\Phi_{\omega}(\hat X)]}.
\end{align}

We study the evolution of two positive operators $\hat X,\hat Y>0$ propagated by the record $\omega$.
Since $\Phi_{\omega}$ is positive, the Hilbert's projective metric is preserved.
It follows directly that
\begin{align}
    d_{H}
    \left(
        \widetilde{\Phi}_{\omega}(\hat X),
        \widetilde{\Phi}_{\omega}(\hat Y)
    \right)=
    d_{H}
    \left(
        \Phi_{\omega}(\hat X),
        \Phi_{\omega}(\hat Y)
    \right)
    \leq
    d_{ H}(\hat X,\hat Y).
    \label{eq:projective}
\end{align}
Thus, an arbitrary measurement outcome cannot increase the Hilbert's projective distance between two states with the same record.
On the other hand,
we call a record block $\omega$ \emph{good}, if its conditional map obeys the positivity bound
\begin{align}
    a \tr(\hat X)\hat I
    \leq
    \Phi_{\omega}(\hat X)
    \leq
    b \tr(\hat X)\hat I,
    \qquad
    \f \hat X\geq0,
	\label{eq:bound}
\end{align}
for some $0<a<b<\infty$ independent of $\hat X$. Such a map $\Phi_{\omega}$ has finite projective diameter and, by Theorem ~\ref{Birkhoff}, we have
\begin{align}
    d_{ H}\left(\widetilde{\Phi}_{\omega}(\hat X),\widetilde{\Phi}_{\omega}(\hat Y) \right)\leq
    q d_{H}(\hat X,\hat Y),
    \qquad \e q<1.
	\label{eq:contraction}
\end{align}
In a word,
we have the contraction property on a good record block.

Let the evolution be divided into consecutive blocks of duration $\tau$, and let $G_n$ denote the event that the record in the $n$th block is good. Denote $N_G(n)$ as the number of good blocks among the first $n$ blocks, repeated use Eqs.~\eqref{eq:projective} and~\eqref{eq:contraction} gives
\begin{align}
    d_{ H}(\hat X_n,\hat Y_n)
    \leq
    q^{N_G(n)}d_{ H}(\hat X_0,\hat Y_0).
 	\label{eq:good_block}
\end{align}
Consequently, to show the exponential contraction, it remains only to show that $\liminf_{n\rightarrow\infty}\frac{N_G(n)}{n}>0$ almost surely.

A sufficient recurrence condition is as follows:
\begin{align}
    \e p_*>0,\,\f n,\,
    \mathbb P
    \left(
        G_n\mid\mathcal F_{n\tau}
    \right)
    \geq p_*.
    \label{eq:condition}
\end{align}
Combining this with Eq.~\eqref{eq:good_block} yields an exponential contraction of the projective distance.
Indeed, we have
\begin{align}
    \mathbb P
    \left(
        N_G(n)\ge \nu n
    \right)
    \geq
    \sum_{m\ge\nu n}C^m_n p_*^m(1-p_*)^{n-m}.
    \label{eq:nogood}
\end{align}
Taking $n\rightarrow\infty$, the right hand side converges to $1$ when $\nu< p_*$, by the law of large numbers applied to binomial distribution.
Hence good blocks occur infinitely at a positive density often almost surely, and Eq.~\eqref{eq:good_block} yields $d_{H}(\hat X_n,\hat Y_n) \le q^{\nu n}d_{ H}(\hat X_0,\hat Y_0)$ almost surely for any $\nu< p_*$. This establishes the exponential contraction for the projective distance.

The corresponding long-time contraction rate can be characterized within the standard Lyapunov framework. Since the linear evolution forms a random cocycle, the multiplicative ergodic theorem~\cite[Theorem~3.4.1]{Ludwig1998} ensures that the associated asymptotic exponent exists almost surely.
For related results on stochastic contraction, see Refs.~\cite{vanHandel2009,Amini2021,Amini2026}.

Next we prove the contraction for two different unravelings.
The essential ingredient of the argument is that we need to prove the set of finite-time records that produces uniform strict contraction (in particular, good record blocks) has a conditional probability bounded away from zero. Any unraveling satisfying these conditions leads to the same almost-sure contraction.

\subsection{Contraction for the number unraveling}
For the number unraveling, let $G_n^{(\mathrm N)}$ denote the event that no jump is detected during the time block $[n\tau,(n+1)\tau]$. The corresponding conditional evolution is the no-count propagator $e^{\mathcal L_{\mathrm{nc}}\tau}$. Since this propagator satisfies the condition in Eq.~\eqref{eq:bound}, every such no-jump block is strictly contractive, and therefore $G_n^{(\mathrm N)}$ is a sub-event of $G_n$.

The conditional jump intensity is $ \nu_t=\Gamma\tr[\hat O^\dagger\hat O\hat\rho_c(t)]$.
Suppose that, along the states relevant to this record block, $\nu_t\leq\nu_+<\infty$ is bounded above (which is true if $\hat{O}$ is a bounded operator, and it is always the case in practical examples).
Then the conditional probability of $G^{(\mathrm N)}_n$ is then bounded from below by
\begin{align}
    \mathbb P
    \left(
        G_n^{(\mathrm N)}
        \mid
        \mathcal F_{n\tau}
    \right)
    \geq
    e^{-\nu_+\tau}
    :=
    p_{\mathrm N}>0.
\end{align}
Therefore, every $G_n^{(\mathrm N)}$ has a uniformly positive conditional probability of occurring in each block. Since $G_n^{(\mathrm N)}\subseteq G_n$, we have
\begin{align}
    \mathbb P
    \left(
        G_n\mid\mathcal F_{n\tau}
    \right)
    \geq
    \mathbb P
    \left(
        G_n^{(\mathrm N)}
        \mid\mathcal F_{n\tau}
    \right)
    \geq
    p_{\mathrm N}>0.
\end{align}
Hence the condition in
Eq.~\eqref{eq:condition} is satisfied.
Hence the contractive events $G_n$ occur infinitely at a positive density often almost surely. Consequently, the quantum-jump trajectories satisfy that for $\nu < p_{\mathrm{N}}$:
\begin{align}
    d_{H}(\hat X_t,\hat Y_t)
    \le Cq^{\nu t} \quad(\text{as}\,\,t\to\infty)
    \qquad
    \text{almost surely}.
\end{align}

\subsection{Contraction for the homodyne unraveling.}
For the diffusive homodyne unraveling, let $W_t$ denote the Wiener process. Choose a continuous reference path $w_*(s)$, $0\leq s\leq\tau$, with $w_*(0)=0$, such that the associated conditional block map satisfies Eq.~\eqref{eq:condition}.
For example, for $w_*(s)=0$, the stochastic contribution vanishes and the corresponding map is $\Phi_{\omega}=e^{(\mathcal{L}_0+\Gamma \mathcal{D}[\hat{\mathcal{O}}])\tau}$ satisfies Eq.~\eqref{eq:bound} due to Eq.~\eqref{eq:bound0}.
Since the finite-time conditional propagator depends continuously on the record, the same uniform contraction remains valid in a sufficiently small neighborhood of this reference path. Take $\varepsilon$ the diameter of this neighborhood, then we can define the good diffusive block
\begin{align}
    G_n^{(\mathrm H)}
    =
    \left\{\sup_{0\leq s\leq\tau}\left| W_{n\tau+s}-W_{n\tau}- w_*(s)\right|<\varepsilon\right\}.
    \label{eq:wiener_block}
\end{align}

Then $G_n^{(\mathrm H)}$ is a sub-event of $G_n$.
Wiener measure has full support on the space of continuous paths starting at the origin. Hence every open tube of the form
Eq.~\eqref{eq:wiener_block} has strictly positive probability,
\begin{align}
    \mathbb P
    \left(
        G_n^{(\mathrm H)}\mid\mathcal F_{n\tau}
    \right)
    =p_{\mathrm H}>0.
    \label{eq:wiener_probability}
\end{align}
Moreover, Wiener increments over disjoint time blocks are independent,
so the same probability $p_{\mathrm H}$ applies to every block.
Hence condition Eq.~\eqref{eq:condition} holds as $G^{(\mathrm H)}_n$ is a sub-event of $G_n$,
which again implies for any $\nu< p_{\mathrm H}$,
we have
\begin{align}
    d_{H}(\hat X_t,\hat Y_t)
    \le C q^{\nu t}
    \quad(\text{as}\,\,t\to\infty)
    \qquad
    \text{almost surely}.
\end{align}

\subsection{Robustness of synchronization across different unravelings}
As discussed above, the stochastic contraction is independent of the ancilla measurement basis used in the weak-measurement construction. Combined with the transitive symmetry of the dynamics, the synchronization effect is also independent of the ancilla measurement basis. We now numerically investigate the synchronization dynamics of the minimal model introduced in the main text under several distinct unravelings of the same master equation. 

\begin{figure}[htbp]
    \begin{centering}
    \includegraphics[width=0.4\linewidth]{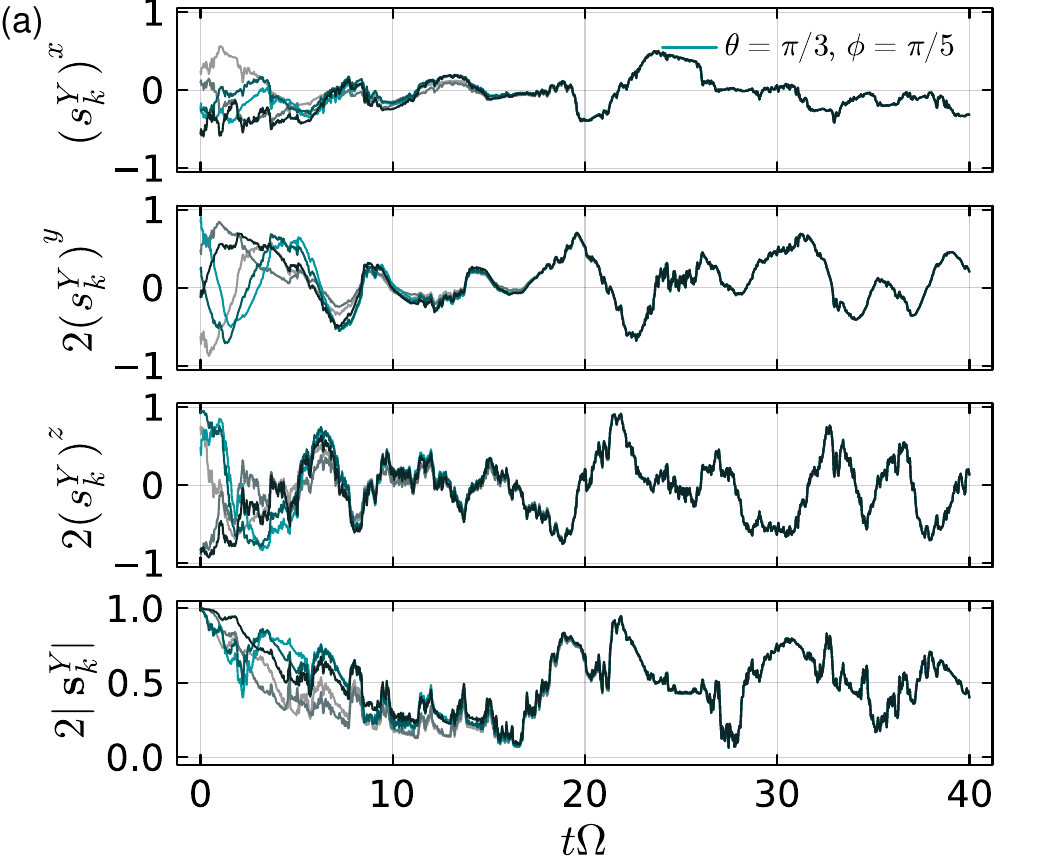}
    \includegraphics[width=0.4\linewidth]{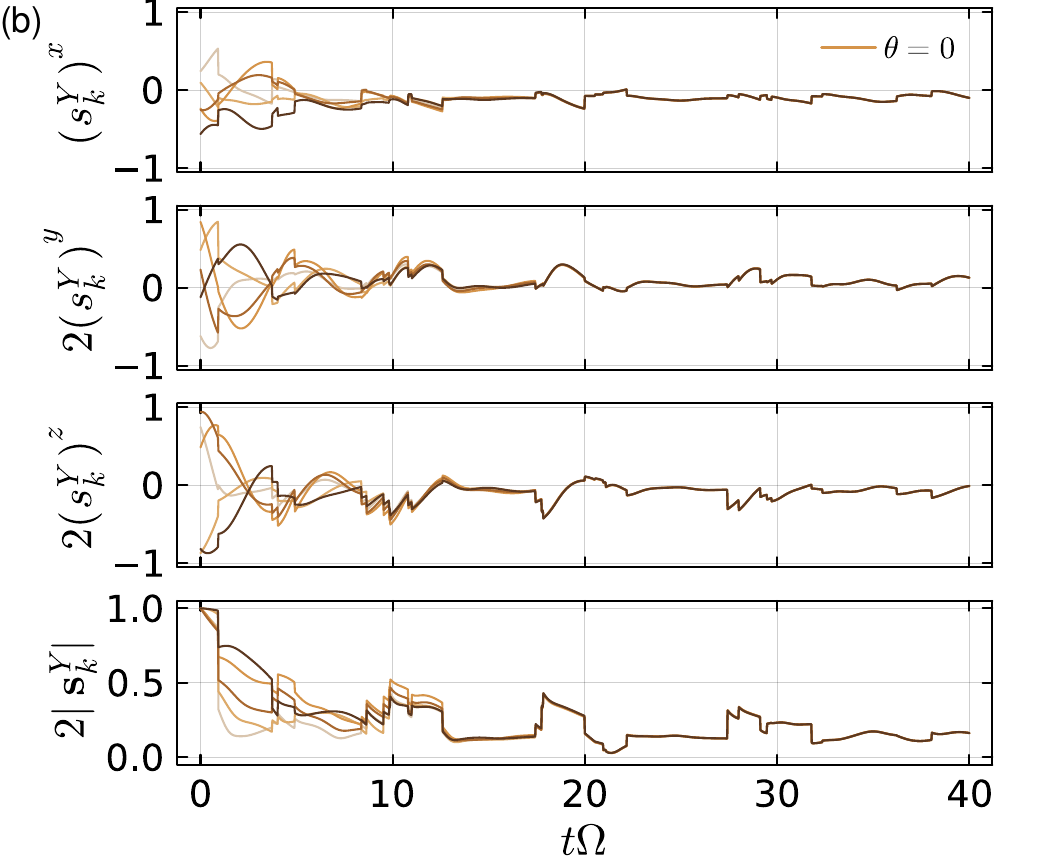}\\
    \par\end{centering}
    \caption{
    Single-trajectory synchronization dynamics under different unravelings. (a) Homodyne unraveling with $\theta=\pi/3$ and $\phi=\pi/5$. (b) Number unraveling with $\theta=0$. Solid lines in different colors represent different spins.
    Other parameters are identical to Fig.1(b) in the main text.}
    \label{Figs1}
\end{figure}

In addition to the homodyne unraveling with $\phi=0$ in the main text, we demonstrate through examples that the single-trajectory synchronization phenomenon exists under other unraveling---specifically, the homodyne unraveling with $\theta=\pi/3$ and $\phi=\pi/5$ (Fig.~\ref{Figs1}(a)), and the number unraveling with $\theta=0$ (Fig.~\ref{Figs1}(b)).
Despite the measurement records are fundamentally different, the mechanism of synchronization holds for both scenarios.

\section{Extensions beyond the minimal model}
In this section, we show the generality of the effect. We take three representative examples to illustrate three complementary extensions. An interacting chain shows that full permutation symmetry can be relaxed to a weaker transitive symmetry. A spin-$1$ ensemble demonstrates that the algebraic condition is not limited to spin-$1/2$ systems. Finally, the Tavis--Cummings model provides a promising implementation. 
These examples exhibit full, anti-, partial, and cluster synchronization.

\subsection{XY spin chain}
We first apply our framework to an $N$-spin nearest-neighbor XY chain in a transverse field. The Hamiltonian is
\begin{align}
\hat H_{\text{XY}} = \Omega \hat S_x+\frac{J}{2} \sum_{k}
\left( \hat\sigma_k^x \hat\sigma_{k+1}^x + \hat\sigma_k^y \hat\sigma_{k+1}^y \right) + \frac{1}{2}\sum_k \Delta_k \hat \sigma_k^z ,
\end{align}
For periodic boundary conditions (PBC), the sum runs over $k=1,\ldots,N$ with $\hat\sigma_{N+1}^\alpha\equiv \hat\sigma_1^\alpha$. For open boundary conditions (OBC), the sum runs over $k=1,\ldots,N-1$.

The dissipative part is consistent with Eq. 1 in the main text. We include unmonitored local spontaneous emission channels $\hat\sigma_k^-$, with $\gamma$ the decay rate, and continuously monitor the collective observable $\hat S^z$ through homodyne detection, with $\Gamma$ the measurement strength. So the conditional state $\hat \rho_{c}$ of the system evolves according to the It\^o stochastic master equation
\begin{align}
d\hat\rho_c
=-i[\hat H_{\text{XY}},\hat\rho_c]dt+ \gamma\sum_k \mathcal{D}[\hat\sigma^-_k]\hat\rho_c dt+\Gamma\mathcal{D}[\hat S^z]\rho_c dt+\sqrt{\Gamma}\mathcal{H}[e^{-i\phi}\hat S^z]\hat\rho_c dW_t, 
\end{align}
where $\mathcal D[\hat L]\hat\rho=\hat L\hat\rho \hat L^\dagger -\frac{1}{2}\{\hat L^\dagger \hat L,\hat\rho\}$, $\mathcal H[\hat O]\hat\rho = \hat O\hat\rho+\hat\rho\hat O^\dagger -\mathrm{Tr}[(\hat O+\hat O^\dagger)\hat\rho]\hat\rho$, $d W_t$ is a Wiener increment.

\begin{figure}[htbp]
    \begin{centering}
	\includegraphics[width=0.4\linewidth]{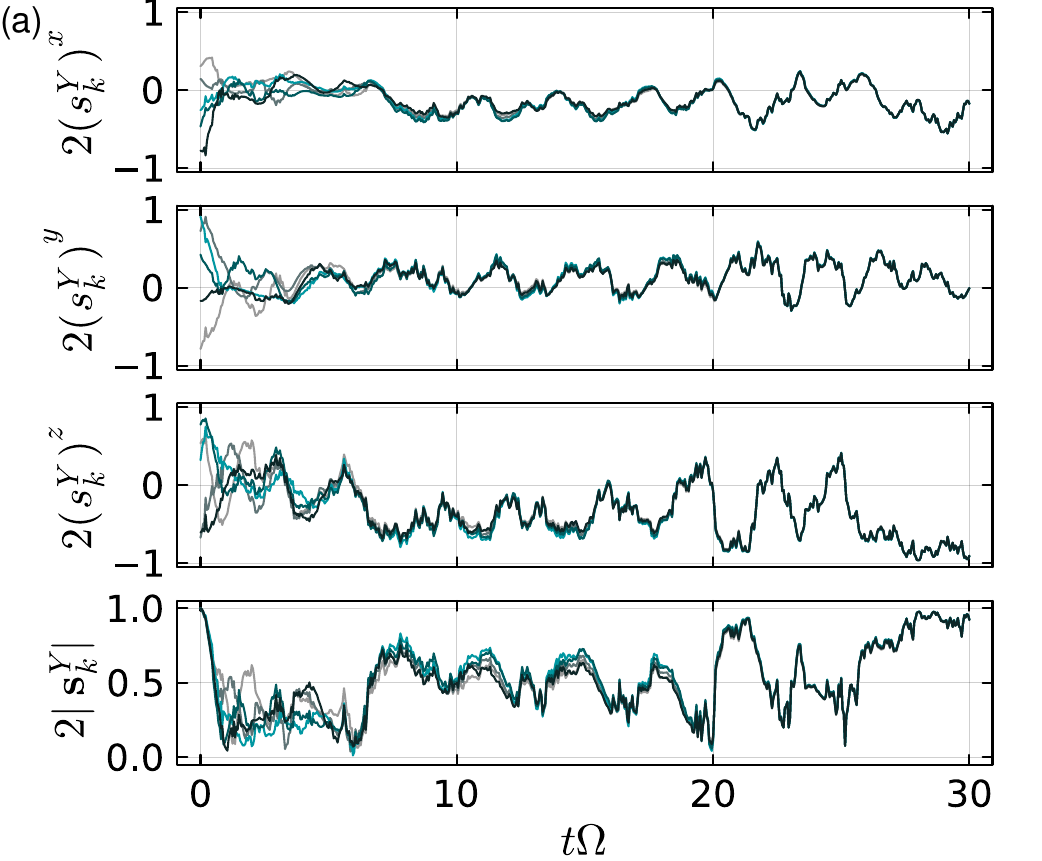}
    \includegraphics[width=0.4\linewidth]{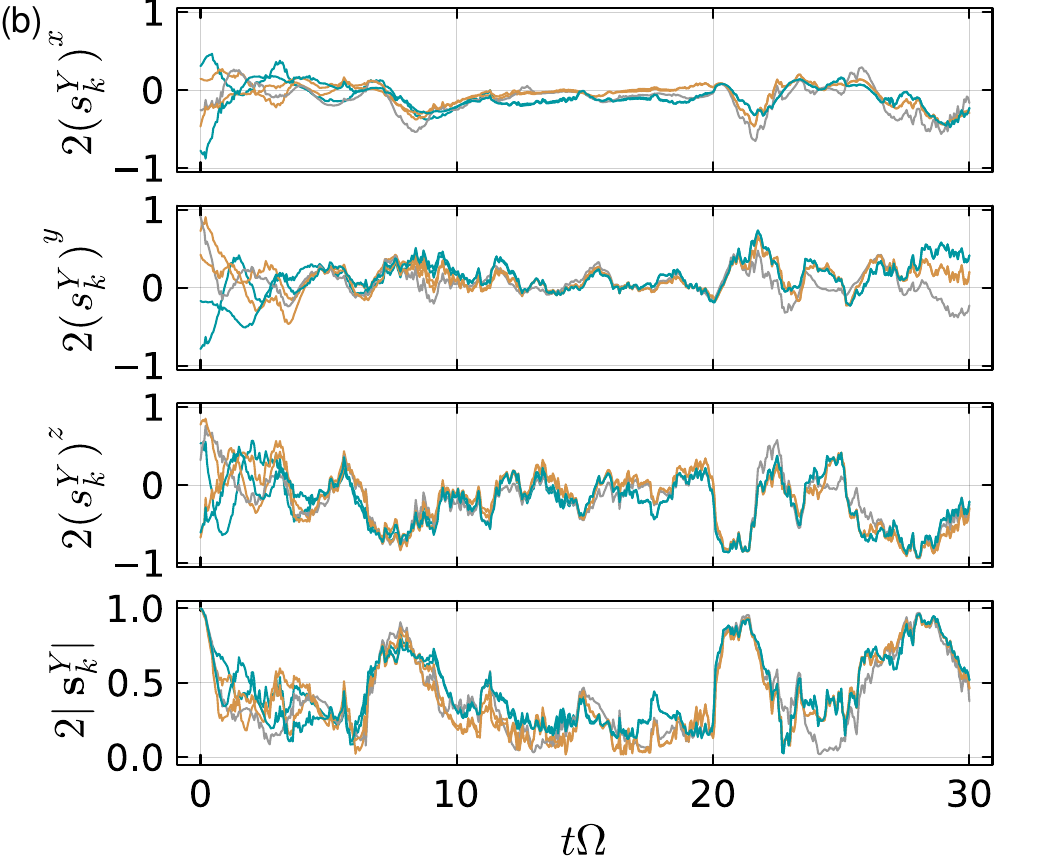}\\
    \par\end{centering}
    \caption{
    Synchronization in the nearest-neighbor XY spin chain.
    (a) Synchronization under PBC. Here $\Delta_k=0.5\Omega$ for all spins.
    (b) Cluster synchronization under OBC. We set $\Delta_k=(0.5,0.2,0.9,0.2,0.5)\Omega$, the synchronization occurs within the
    reflection-related pairs: spins $(1,5)$ in green and spins $(2,4)$ in orange, while the spin $3$ is shown in gray. Other parameters: 
    $N=5$, $J=\Omega$, $\Gamma=0.8\Omega$, $\gamma=0.1\Omega$, $\phi=0$.
    }
    \label{Figs2}
\end{figure}

For PBC,  we take $\Delta_k=\Delta$, the system has a translation symmetry. Let $\hat P_{12\cdots N}$ denote the cyclic translation operator that maps the sites $(1,2,\ldots,N)$ to $(2,3,\ldots,N,1)$, and define the corresponding action on an operator $\hat X$ as $\mathcal P_{12\cdots N}[\hat X] := \hat P_{12\cdots N}\hat X \hat P_{12\cdots N}^{\dagger}$.
Since the Hamiltonian, the local dissipative channels, and the collective measurement are invariant under this cyclic translation, the conditioned dynamics satisfies $\hat\rho_c^Y\left(t;\mathcal P_{12\cdots N}[\hat\rho_0]\right)=\mathcal P_{12\cdots N}\left[\hat\rho_c^Y(t;\hat\rho_0)\right]$.
Thus, although the model is not invariant under arbitrary permutations, the remaining translation symmetry still relates all sites to one another. Together with the full single-site algebra $\mathcal B(\mathbb C^2)$ generated on each spin, the synchronization mechanism remains valid. The numerical results are shown in Fig.~\ref{Figs2}.

As a comparison, we also consider OBC. In this case, the translation symmetry is broken, and not all spins can be related by a symmetry operation. Nevertheless, the open chain still retains a reflection symmetry, which maps site $k$ to site $N+1-k$. This also can be viewed as a partial breaking of the full permutation symmetry. The synchronization mechanism then applies only within the remaining symmetry-related pairs. For example, in Fig.~\ref{Figs2}, synchronization therefore occurs within the pairs $(1,5)$ and $(2,4)$, but not necessarily between different pairs.

\subsection{Spin-1 ensemble}
We next consider an ensemble of spin-$1$ particles. This example shows that the algebraic origin of the synchronization is not restricted to two-level local Hilbert spaces. Instead of Pauli operators, we use the spin-$1$ operators $\hat F_k^\alpha$ on site $k$, The local Hilbert space $\mathbb C^3$ is spanned by the Zeeman states $\ket{+}_k,\,\ket{0}_k,\,\ket{-}_k$. In this basis, the spin-$1$ matrices are
\begin{align}
F^z=
\begin{pmatrix}
1&0&0\\
0&0&0\\
0&0&-1
\end{pmatrix},
\qquad
F^x=\frac{1}{\sqrt2}
\begin{pmatrix}
0&1&0\\
1&0&1\\
0&1&0
\end{pmatrix},
\qquad
F^y=\frac{1}{\sqrt2}
\begin{pmatrix}
0&-i&0\\
i&0&-i\\
0&i&0
\end{pmatrix}.
\end{align}

We consider the spin-$1$ Hamiltonian
\begin{align}
\hat H_{\text{spin1}} = \Omega \hat J^x + \Delta \hat J^z
+ q\sum_{k=1}^{N}(\hat F_k^z)^2,
\end{align}
with $\hat J^\alpha=\sum_{k=1}^{N}\hat F_k^\alpha,\,\alpha=x,y,z$.

\begin{figure}[htbp]
    \begin{centering}
    \includegraphics[width=0.4\linewidth]{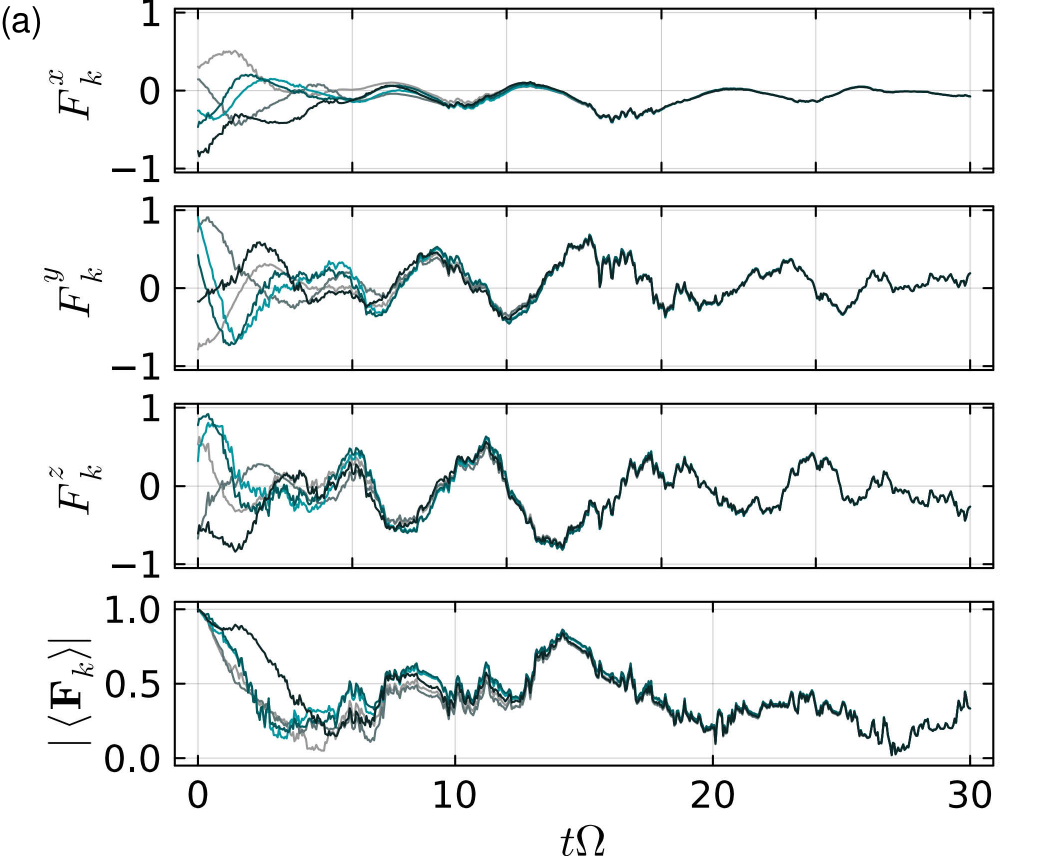}
    \includegraphics[width=0.4\linewidth]{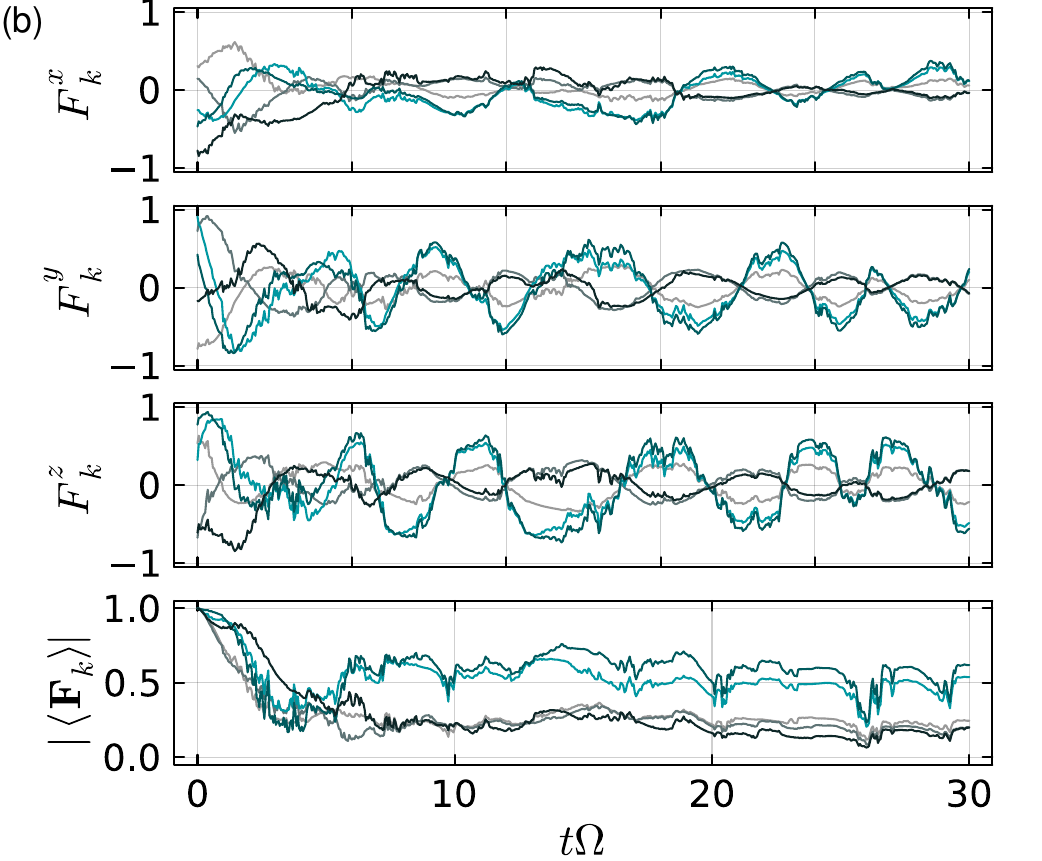}\\
    \par\end{centering}
    \caption{
    Synchronization in an ensemble of noninteracting spin-$1$ particles. 
    (a) Complete synchronization in the presence of local decay,
    $\gamma_{+0}=\gamma_{0-}=0.1\Omega$.
    (b) Reduced collinear synchronization in the absence of local decay.
    Other parameters: $N=5$, $\Delta=0.5\Omega$, $\Gamma=0.8\Omega$, $q=0.2\Omega$.
    }
    \label{Figs3}
\end{figure}

We include unmonitored local relaxation processes. For this spin-$1$ system with a quadratic Zeeman term, the most natural choice is to resolve the two adjacent transitions in the Zeeman ladder \(|+\rangle_k\rightarrow |0\rangle_k\rightarrow |-\rangle_k\). Accordingly, we define the local jump operators as $L_{k,+0} = |0\rangle_k\langle+|$ and $L_{k,0-} = |-\rangle_k\langle0|$, with corresponding dissipation rates $\gamma_{+0}$ and $\gamma_{0-}$. Then we continuously monitor the collective longitudinal magnetization $\hat{J}^z$ with the strength $\Gamma$. In this case, the conditional dynamics is described by
\begin{align}
    d\hat\rho_c = &-i[\hat H_{\text{spin1}},\hat\rho_c]dt + \sum_{k}\left(\gamma_{+0}\mathcal{D}[\hat L_{k,+0}]\hat\rho_c+\gamma_{0-}\mathcal{D}[\hat L_{k,0-}]\hat\rho_c\right)\nonumber\\
    &+\Gamma\mathcal D[\hat J^z]\hat\rho_c dt + \sqrt{\Gamma}\mathcal H[\hat J^z]\hat\rho_c dW_t.
    \label{eq:SME-spin1}
\end{align}
The local dissipations, together with the local coherent Hamiltonian, generate the von Neumann algebra $\mathcal B(\mathbb C^3)$
for each spin. Therefore, the irreducibility condition for the QMS is satisfied, and dynamical contraction holds. Numerically, we compute the evolution of the spin vectors
$\mathbf f_k(t) =\left(\langle F_k^x\rangle, \langle F_k^y\rangle, \langle F_k^z\rangle\right)$, as shown in Fig.~\ref{Figs3}(a). We also consider the partial collinear synchronization, shown in Fig.~\ref{Figs3}(b), by setting $\gamma_{+0}=\gamma_{0-}=0$. In this scenario the local decay channels are absent, the operators entering the QMS no longer generate the full algebra on the whole many-body Hilbert space. Instead, the generated algebra acts fully only within a subspace. This is similar to the discussion regarding spin-$1/2$ system in the main text.

\subsection{Tavis--Cummings realization}
We expect to explore this trajectory-level synchronization effect on new experimental platforms that is capable of simultaneous collective and local monitoring. As a potential realization, we consider the Tavis-Cummings model implemented using superconducting circuit-QED, in which several superconducting qubits are coupled to a common microwave resonator.

We denote the resonator annihilation operator by $\hat a$, and the lowering and raising operators of qubit $k$ by $\hat\sigma_k^-$ and $\hat\sigma_k^+$, respectively. In a frame rotating at the drive frequency, the driven Tavis--Cummings Hamiltonian reads
\begin{align}
\hat H_{\rm TC} = \Delta_c \hat a^\dagger \hat a + \Omega \hat S^x + \Delta \hat S^z + g\left( \hat a^\dagger \hat S^- + \hat a \hat S^+ \right),
\end{align}
where $\Delta_c$ is the cavity detuning, $\Delta$ is the qubit detuning, $\Omega$ is the common transverse Rabi drive, and $g$ is the homogeneous qubit-resonator coupling strength.

\begin{figure}[htbp]
    \begin{centering}
    \includegraphics[width=0.4\linewidth]{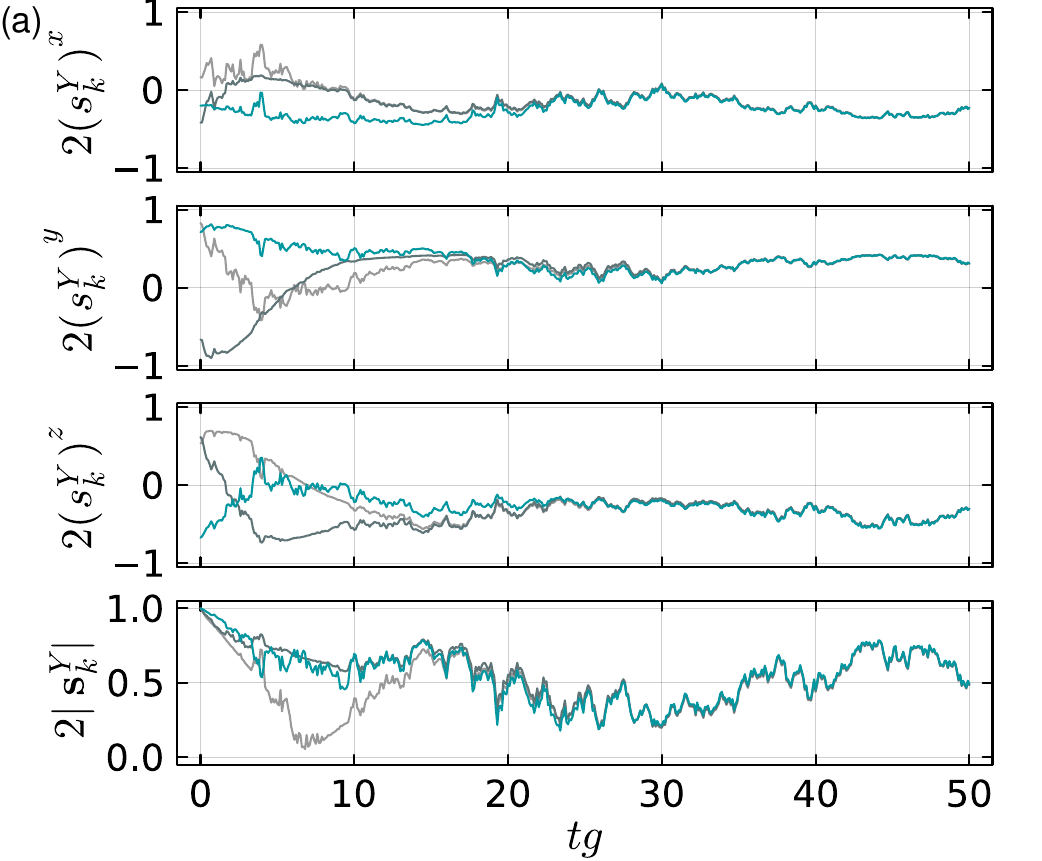}
    \includegraphics[width=0.4\linewidth]{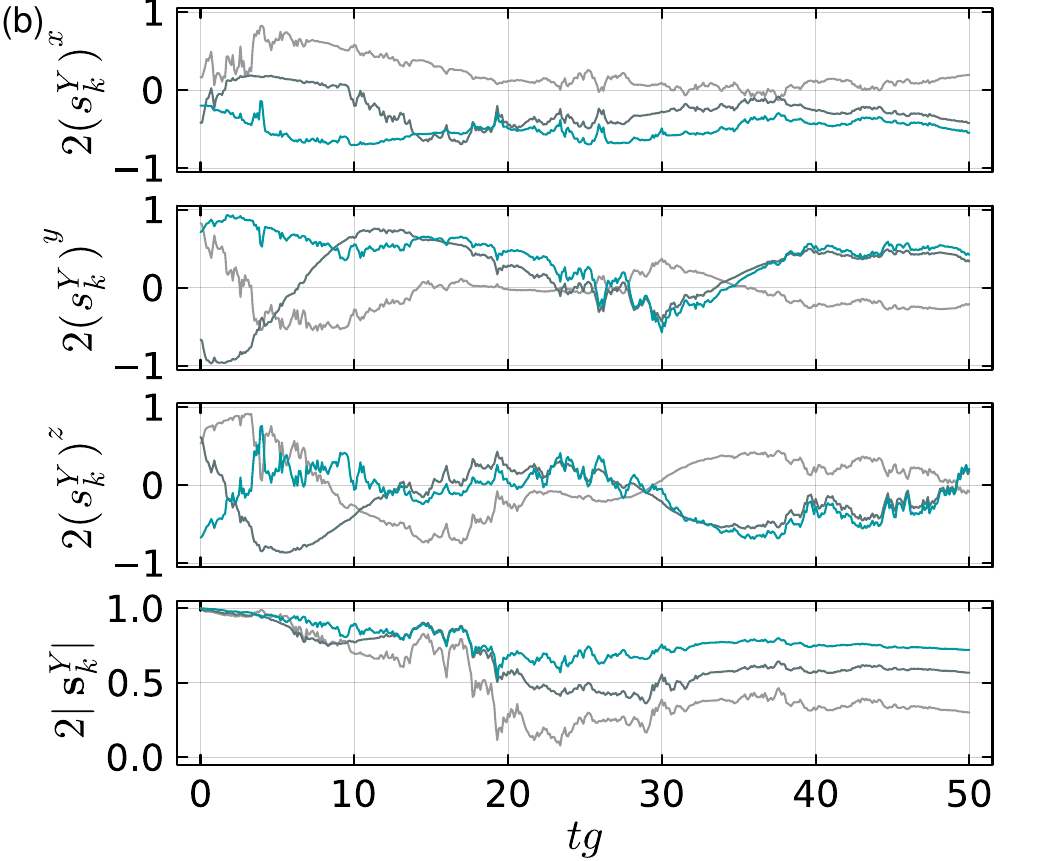}
    \par\end{centering}
\caption{
Synchronization in the Tavis--Cummings model. (a) Complete synchronization in the presence of local dissipation.
(b) Reduced collinear synchronization without local dissipation. 
}
    \label{Figs4}
\end{figure}

Additionally, each qubit undergoes homogeneous local relaxation $\hat\sigma_k^-$ and pure dephasing $\hat\sigma_k^z$, with corresponding dissipation rates $\gamma_{-}$ and $\gamma_{\phi}$.
The common cavity output is monitored by homodyne detection, giving the conditioned stochastic master equation
\begin{align}
d\hat\rho_c
=&
-i[\hat H_{\rm TC},\hat\rho_c] dt +\sum_{k}\left(\gamma_{-}\mathcal D[\hat\sigma_k^-]\hat\rho_c +\gamma_{\phi}\mathcal D[\hat\sigma_k^z]\hat\rho_c \right)dt\nonumber\\&+ \kappa\mathcal D[\hat a]\hat\rho_c dt+\sqrt{\eta\kappa}\mathcal H[e^{-i\phi}\hat a]\hat\rho_c dW_t .
\end{align}

For the numerical benchmark, we use the experimentally motivated scales summarized in the parameter audit: $N=3$, $g/2\pi=3.6\,{\rm MHz}$, $\kappa/2\pi=43\,{\rm MHz}$, $|\Delta_c|/2\pi=25\,{\rm MHz}$, $\gamma_-/2\pi=0.041\,{\rm MHz}$, $\gamma_\phi/2\pi=0.26\,{\rm MHz}$, and measurement efficiency $\eta=0.72$~\cite{Mlynek2014}. The transverse drive is set to $\Omega/2\pi=1\,{\rm MHz}$, with $\Delta=0$, $\phi=0$, and no direct cavity drive.

This Tavis-Cummings model implements the same essential structure as the minimal spin model, but the collective measurement is mediated by a dynamical cavity mode. In the presence of local channels, the calculated results in Fig. \ref{Figs4}(a) demonstrate complete synchronization. When the local channels are removed, as shown in Fig. \ref{Figs4}(b), the local spin polarizations of the qubits exhibit partial directional synchronization.

\section{Proof of reduced collinear synchronization}
In this section we clarify why the collective-only dynamics leads to reduced collinear synchronization rather than complete synchronization. 
The key point is that, without local dissipation, the conditioned dynamics acts only on the collective spin representation spaces, while the multiplicity spaces remain unresolved.

\subsection{Schur--Weyl decomposition}
We consider the collective-only limit, $\gamma=0$. The Hilbert space of $N$ spin-$1/2$ particles carries two mutually commuting actions: the collective $SU(2)$ action generated by $\hat S^\alpha=\frac{1}{2}\sum_{k=1}^N\hat\sigma_k^\alpha,\,\alpha=x,y,z ,$ and the action of the permutation group $S_N$ that permutes the spin labels.
Schur--Weyl duality gives the decomposition $\mathcal H = (\mathbb C^2)^{\otimes N} = \bigoplus_S \mathcal V_S\otimes \mathcal M_S$ as $SU(2)\times S_N$-representations~\cite{Rowe2012,Chase2008}.
Here $\mathcal V_S$ is the spin-$S$ irreducible representation of $SU(2)$, with dimension $2S+1$, and $\mathcal M_S$ is the corresponding multiplicity space, which carries the degeneracy of equivalent spin-$S$ representations.
The allowed values of $S$ are $S=\frac{N}{2},\frac{N}{2}-1,\ldots ,$ down to $0$ for even $N$ and $1/2$ for odd $N$.

Equivalently, one may choose basis states $\ket{S,\mu_s,\beta}$, where $\mu_s=-S,-S+1,\ldots,S$, $\beta=1,\ldots,m_S$ and $m_S=\dim\mathcal M_S$ is the multiplicity of the spin-$S$ representation.
In this basis, the collective spin operators do not change the multiplicity label $\beta$. Therefore $\hat S^\alpha = \bigoplus_S \hat S_S^\alpha\otimes \hat I_{\mathcal M_S},\,\alpha=x,y,z$.
Here $\hat S_S^\alpha$ is the spin-$S$ matrix acting on $\mathcal V_S$, and $\hat I_{\mathcal M_S}$ is the identity operator on $\mathcal M_S$. Thus collective operators act nontrivially only on $\mathcal V_S$ and are blind to the multiplicity space.

For $N$ spin-$1/2$ particles, the multiplicity is
\begin{align}
m_S
=\dim\mathcal M_S
=\binom{N}{\frac{N}{2}-S}-\binom{N}{\frac{N}{2}-S-1},
\label{eq:ms}
\end{align}
with the convention that $\binom{N}{r}=0$ for $r<0$ or $r>N$. 
The number of product states with total $S^z=\mu$ is $\binom{N}{N/2-\mu}$. On the other hand, a spin-$S'$ irreducible representation contains exactly one state with $\mu$ whenever $S'\geq |\mu|$. Therefore $\binom{N}{\frac{N}{2}-\mu} = \sum_{S'\geq |\mu|} m_{S'}$, where $m_{S'}=\dim\mathcal M_{S'}$ is the multiplicity of the spin-$S'$ representation. Taking the difference between the cases $\mu=S$ and $\mu=S+1$ gives Eq.~\eqref{eq:ms}.

\subsection{Synchronization within a frozen total-spin sector}
We now explain why the collective-only dynamics lead to reduced collinear synchronization. In the presence of a strong symmetry, a single conditioned trajectory asymptotically selects one symmetry sector~\cite{Benoist2014,SanchezMunoz2019}. Here these sectors are the total-spin sectors. Thus, for almost every trajectory, there exists a total spin $S_*$ such that
\begin{align}
P_S^Y(t)=\tr[\hat\Pi_S\hat\rho_c^Y(t)]
\longrightarrow
\delta_{S,S_*},
\end{align}
where $\hat\Pi_S$ projects onto $\mathcal V_S\otimes\mathcal M_S$. 
That is known as the dissipative freezing~\cite{SanchezMunoz2019}.
After this freezing, the conditional dynamics is restricted to the selected $S_*$-sector.
Consequently, the conditioned state asymptotically factorizes as
\begin{align}
\hat\rho_c^Y(t)
\xrightarrow{t\rightarrow\infty}
\frac{\hat\rho_{\mathcal V}^Y(t)\otimes \hat\rho_{\mathcal M}^Y}{\tr[\hat\rho_{\mathcal V}^Y(t)]\,\tr[\hat\rho_{\mathcal M}^Y]}.
\end{align}
Here $\hat\rho_{\mathcal V}^Y(t)\otimes \hat\rho_{\mathcal M}^Y$ denotes the normalized conditioned state after restriction to the selected sector $\mathcal V_{S_*}\otimes\mathcal M_{S_*}$. The component $\hat\rho_{\mathcal V}^Y(t)$ acts on the spin representation space $\mathcal V_{S_*}$, while $\hat\rho_{\mathcal M}^Y$ acts on the multiplicity space $\mathcal M_{S_*}$.
By the Schur--Weyl decomposition, the collective spin operators take the form $\hat\Pi_{S_*}\hat S^\alpha\hat\Pi_{S_*} = \hat S_{S_*}^\alpha\otimes \hat I_{\mathcal M_{S_*}},\,\alpha=x,y,z$. 
Therefore, the collective dynamics acts only on the spin representation space $\mathcal V_{S_*}$ and does not resolve the multiplicity space $\mathcal M_{S_*}$.
On $\mathcal V_{S_*}$, the stochastic contraction holds,
as the restricted Hamiltonian $\hat H_{S_*} = \Omega \hat S^x_{S_*}+\Delta \hat S^z_{S_*}$ together with the monitored operator $\hat S^z_{S_*}$ generates the full operator algebra $\mathcal B(\mathcal V_{S_*})$ on $\mathcal V_{S_*}$. 
Hence, along a fixed measurement record $Y$, for two initial states $\hat\rho_0$ and $\hat\rho_0'$ that select the same total-spin sector,
\begin{align}
\norm{\frac{\hat\rho_{\mathcal V}^Y(t;\hat\rho_{0})}{\tr[\hat\rho_{\mathcal V}^Y(t;\hat\rho_{0})]}-\frac{\hat\rho_{\mathcal V}^{Y}(t;\hat\rho_0^\prime)}{\tr[\hat\rho_{\mathcal V}^{Y}(t;\hat\rho_0^\prime)]}}_1
\xrightarrow{t\rightarrow\infty} 0.
\end{align}
Note that the contraction does not extend to the full space $\mathcal V_{S_*}\otimes\mathcal M_{S_*}$, because the collective operators are blind to $\mathcal M_{S_*}$.  The stochastic contraction occurs only in $\mathcal V_{S_*}$.

In the fixed $S_*$-sector, by Wigner--Eckart theorem~\cite{Sakurai2020},
the diagonal block of a local Pauli operator separates into a spin part and a multiplicity part, i.e. $\hat\Pi_{S_*}\hat\sigma_k^\alpha\hat\Pi_{S_*} = \hat S_{S_*}^\alpha\otimes \hat A_{k,S_*},\,\alpha=x,y,z$, where $\hat A_{k,S_*}$ acts only on $\mathcal M_{S_*}$. 
Taking the expectation values of Pauli operators gives
\begin{align}
s_k^{\alpha,Y}(t)\xrightarrow{t\rightarrow\infty}\frac{1}{2}\tr[\hat\rho_{c}^Y(t)\hat\sigma_k^\alpha]=\frac{1}{2}\frac{\tr[\hat\rho_{\mathcal V}^Y(t)\hat S_{S_*}^\alpha]}{
\tr[\hat\rho_{\mathcal V}^Y(t)]}\frac{\tr[
\hat\rho_{\mathcal M}^Y\hat A_{k,S_*}
]}{\tr[\hat\rho_{\mathcal M}^Y]}.
\end{align}
Equivalently, the local Pauli polarization vector takes the asymptotic form
\begin{align}
\mathbf s_k^Y(t) \xrightarrow{t\rightarrow\infty} a_k^Y\mathbf v^Y(t), 
\end{align}
where $a_k^Y=\frac{\tr[\hat\rho_{\mathcal M}^Y\hat A_{k,S_*}]}{2\tr[\hat\rho_{\mathcal M}^Y]}$ is the site-dependent scalar amplitude, and $\mathbf v^Y(t)=\frac{ \tr[\hat\rho_{\mathcal V}^Y(t)\hat {\mathbf{S}}_{S_*}]}{\tr[\hat\rho_{\mathcal V}^Y(t)]}$ is the common trajectory-dependent collective vector for all sites. 
This establishes reduced collinear synchronization. The local spin-polarization vectors share a common collective vector and therefore become collinear, while their amplitudes need not become identical. Since the multiplicity amplitudes $a_k^Y$ are fixed after sector selection, the amplitude ratios approach fixed values,
i.e. 
\begin{align}
\frac{\norm{\mathbf s_i^Y(t)}}{\norm{\mathbf s_j^Y(t)}}
\xrightarrow{t\rightarrow\infty}
\frac{|a_i^Y|}{|a_j^Y|}.
\end{align}

\end{document}